\documentclass[12pt, a4paper, left = 2.5cm, right = 2cm, top = 3cm, bottom = 3cm]{SimpleNotes}

\usepackage{natbib}
\usepackage{float}
\usepackage{tikz}
\usepackage[titletoc,toc,title,page]{appendix}
\usetikzlibrary{shapes,arrows}
\usepackage{amsthm}

\title{Spatial Latent Gaussian Modelling with Change of Support}
\author{
  E.A. Chacón-Montalván\textsuperscript{1},
  P.M. Atkinson\textsuperscript{2},
  C. Nemeth\textsuperscript{3},
  B.M. Taylor\textsuperscript{4},
  P. Moraga\textsuperscript{1}
}

\institute{
  \textsuperscript{1}\textit{Computer, Electrical and Mathematical Sciences and Engineering Division, King Abdullah University of Science and Technology (KAUST), Thuwal 23955-6900, Saudi Arabia}\\
  \textsuperscript{2}\textit{Lancaster Environment Centre, Lancaster University, United Kingdom}\\
  \textsuperscript{3}\textit{Department of Mathematics and Statistics, Lancaster University, United Kingdom}\\
  \textsuperscript{4}\textit{School of Mathematical Sciences, University College Cork, Ireland }\\
}
\date{March 13, 2024}

\usepackage{bm} 
\usepackage{dsfont} 
\usepackage{amsmath} 
\usepackage{amssymb} 

\newcommand{\ssp}{W} 

\newcommand{\rv}[1]{#1} 
\newcommand{\rve}[1]{\bm{#1}} 
\newcommand{\gp}{\text{GP}} 
\newcommand{\gmrf}{\text{GMRF}} 

\newcommand{\ve}[1]{\bm{#1}} 
\newcommand{\m}[1]{\bm{#1}} 

\newcommand{\abs}[1]{\left| #1 \right|} 
\newcommand{\tr}{\intercal}

\newcommand{\cross}[2][]{#2^\tr\ifthenelse{\isempty{#1}}{}{#1} #2}
\newcommand{\tcross}[2][]{#2\ifthenelse{\isempty{#1}}{}{#1} #2^\tr}

\newcommand{\ind}[1]{\mathds{1}_{\left(#1\right)}} 

\newcommand{\real}{\mathbb{R}} 

\usepackage{xifthen}
\newcommand{\mean}[2][]{\mathds{E}\ifthenelse{\isempty{#1}}{}{_{#1}}\left[#2\right]}
\newcommand{\meanhat}[2][]{\hat{\mathds{E}}\ifthenelse{\isempty{#1}}{}{_{#1}}\left[#2\right]}
\newcommand{\var}[2][]{\mathds{V}\ifthenelse{\isempty{#1}}{}{_{#1}}\left[#2\right]}
\newcommand{\pre}[2][]{\text{Prec}\ifthenelse{\isempty{#1}}{}{_{#1}}\left[#2\right]}
\newcommand{\cov}[3][]{\text{Cov}\ifthenelse{\isempty{#1}}{}{_{#1}}\left[#2,#3\right]}
\newcommand{\cor}[3][]{\text{Cor}\ifthenelse{\isempty{#1}}{}{_{#1}}\left[#2,#3\right]}
\newcommand{\pr}[1]{\text{Pr}\left(#1\right)} 
\newcommand{\prd}[1]{\pi\left(#1\right)} 

\newcommand{\df}[1]{\mathcal{#1}} 

\newcommand{\eq}{\text{Equation }} 
\newcommand{\se}{\text{Section }} 
\newcommand{\fig}{\text{Figure }} 

\newcommand{\locs}{\mathbf{s}} 
\newcommand{\locsset}{\mathcal{S}} 

\newcommand{\loc}{l} 
\newcommand{\locsubset}{L} 
\newcommand{\locset}{\mathcal{L}} 

\newcommand{\locc}{c} 
\newcommand{\loccset}{\mathcal{C}} 

\newtheorem{definition}{Definition}
\newtheorem{theorem}{Theorem}
\newtheorem{corollary}{Corollary}

\begin{document}
\maketitle


\begin{abstract}[Data integration, Gaussian processes, geo-additive
  models, land suitability, model-based geostatistics, spatial misalignment.]
  Spatial data are often derived from multiple
  sources (e.g. satellites, in-situ sensors, survey samples) with different supports, but
  associated with the same properties of a spatial phenomenon of interest.
  It is common for predictors to also be measured on different spatial supports than the variables making up the response. Although there is no standard way to work with spatial data with different
  supports, a prevalent approach used by practitioners has been to use
  downscaling or interpolation to project all the variables of analysis towards a common
  support, and then using standard spatial models on this common support. The main disadvantage with this approach
  is that simple interpolation can introduce biases and, more importantly, the uncertainty
  associated with the change of support is not taken into account in the parameter
  estimation of the main model of interest. In this article, we propose a Bayesian
  \textit{spatial latent Gaussian model} that can handle data with different rectilinear supports in the
  response variable and the predictors. Our approach allows us to handle changes
  of support more naturally according to the properties of the spatial stochastic process
  being used, and to take into account the uncertainty from the change of support in
  parameter estimation and prediction. We use spatial stochastic processes
  as linear combinations of basis functions where Gaussian
  Markov random fields define the weights. This process can be projected to different
  supports whilst maintaining the same parameters. Our hierarchical modelling approach
  can be described by the following steps: (i) define a latent model where response
  variables and predictors are considered as latent stochastic processes with continuous
  support, (ii) link the continuous-index set stochastic processes with its projection to
  the support of the observed data, (iii) link the projected process with the observed
  data. We show the applicability of our approach by simulation studies and modelling the
  land suitability of improved grassland in Rhondda Cynon Taf, a county borough in Wales.
\end{abstract}



\section{Introduction}
\label{sec:introduction}



Many research questions in environmental science and public health necessitate the
utilisation of heterogeneous spatial data encompassing different supports. The
characteristics of the support, commonly referred to as spatial sampling units, can vary
significantly across data sources, including variations in size, shape, spacing, and
extent \citep{dungan2002balanced}. In disease incidence and prevalence modelling, it is
common to obtain case data at multiple levels of aggregation, such as at the individual level,
census areas, counties, and districts, along with predictors observed at both individual
and aggregated levels, including satellite pixels \citep{wang2018generalized,
alegana2016advances, lee2022delivering}. In environmental modelling, there is also a
frequent need to combine information from data sources with different supports across both the response and predictor variables, with a particular emphasis on the use of satellite products and synthetic
data derived from climate models \citep{brown2022agent, pacifici2019resolving,
ma2020spatiotemporal}. Integrating different data sources to address a research question
offers the potential to enhance parameter estimation and improve prediction accuracy
\citep{wang2018generalized, pacifici2019resolving, law2018variationala,
leopold2006accounting}. However, to achieve reliable results, it is crucial to employ a
data fusion framework that appropriately acknowledges the support of the data and incorporates
reasonable assumptions to mitigate biases and accurately quantify uncertainty
\citep{gotway2002combininga, pacifici2019resolving}.


The process of
projecting observed data to a common support, known as \textit{change of support} (COS) in
spatial statistics, is essential when integrating heterogeneous data with different supports. This encompasses two key situations: the \textit{change
of support problem} (COSP), where the response variable is observed at different supports,
and \textit{spatial misalignment}, where the support of predictors and responses differ
\citep{gelfand2001change, gotway2002combininga, zhu2003hierarchical}. Both scenarios
require addressing the challenge of effectively utilising the information from different
sources and aligning the data appropriately. In this paper, we adopt the term
\textit{change of support} to encompass both situations, highlighting the importance of
expressing the projection of spatial data to a specific support within a unified framework
for data integration and handling spatial misalignment.


Numerous attempts have been made to address the change of support problem in order to integrate
multiple spatial data sources. The first set of approaches focuses on geostatistical
methods, including point-to-point, block-to-block, point-to-block, and block-to-point
Kriging \citep{kyriakidis2004geostatistical}. These methods involve computing the
covariance matrix using Monte Carlo integration and parameter estimation through variogram
fitting. When dealing with spatial misalignment, predictors are
projected to the response support using one of the types of Kriging and incorporated into
regression models \citep{young2009assessing}.

The second set of approaches utilises Gaussian processes, a widely applied tool in spatial
statistics, to tackle the change of support problem. Several strategies exist within this framework.
One common approach is similar to Kriging, where the covariance matrix is computed using
Monte Carlo integration \citep{gelfand2001change}. Inference under the Gaussian process
framework typically involves optimising the likelihood function or utilising Bayesian
inference. Another strategy involves incorporating spatial point-level auxiliary
variables, approximating the integral of the Gaussian process over a region by averaging
the spatial auxiliary variables within that region \citep{cowles2009reparameterized}. The
objective here is to obtain the posterior distribution of these auxiliary variables.
Additionally, a multivariate approach can be employed, treating all variables as responses
and defining a full covariance matrix between sources. This approach leverages separable
covariances to reduce computational costs and approximates the covariance structure on
aggregated supports using Monte Carlo integration \citep{finley2014bayesian}. Predictions
can be made for all responses or specific ones. An alternative spectral approach, proposed
by \citet{reich2014spectral}, utilises the spectral representation of a Gaussian process.
The model at aggregated support resembles a linear model with pseudo-predictors derived
from spectral analysis. Finally, a computationally efficient approach involves utilising
stochastic partial differential equations (SPDE) to approximate the aggregated Gaussian
process \citep{moraga2017geostatistical}. The representation of the aggregated GP using
SPDE has a similar structure as the continuous GP with a custom linear transformation.
With this approximation, the methods for spatial inference and prediction using SPDE can
be applied.

The third set assumes that data arise from processes that are piecewise
constant on a predefined grid. This allows for analytical simplification, where integrals
become linear combinations of areas. The projection of piecewise
processes to other supports is computed as weighted averages of the latent process values
at cells intersecting the area of interest. These approaches often utilise grids based on
the highest observed resolution or custom computational grids. The main objective is to
obtain the posterior distribution of the piecewise constant process
\citep{taylor2015bayesian,taylor2018continuous,bradley2016bayesian}.


Although various approaches have been developed to address the change of support problem using
Kriging, Gaussian processes, and piecewise processes, they are not without limitations.
Many of these approaches rely on approximating covariance matrices on aggregated supports
using Monte Carlo integration or by averaging spatial random variables. However, the
accuracy of these approximations depends on the number of points used and the sampling
algorithms employed, and the discussion of approximation errors is often neglected.
Moreover, these approximations are necessary even in simpler domains like regular and
rectilinear grids. Another limitation is the lack of a precise connection between the
continuous and aggregated spatial processes in some approaches. Lastly, most existing
approaches are tailored to specific types of models, which restricts their flexibility in
handling more complex modelling scenarios within a unified framework.


In this article, we propose a Bayesian \textit{latent Gaussian spatial model} that
addresses the challenge of handling data with different supports in the response variable
and predictors. Our approach offers a more natural and flexible way to handle the change of
support problem, taking into account the properties of the underlying spatial stochastic process.
We incorporate the uncertainty associated with the change of support into parameter
estimation and prediction. Our proposed model utilises a spatial stochastic process
expressed as a linear combination of basis functions, where the weights are determined by
Gaussian Markov random fields. Notably, this process allows for accurate projection onto
rectilinear supports while preserving the same set of parameters. The hierarchical nature
of our model involves the following steps: (i) defining a latent model using latent
stochastic processes with continuous support that are related to the response variables
and predictors, (ii) establishing the connection between the continuous-index set
stochastic processes and their projection onto the support of the observed data, and (iii)
linking the observed data as noisy realisations of the projected processes. We demonstrate
the practical application of our approach in modelling the land suitability of improved
grassland in Rhondda Cynon Taf, a county borough in Wales. Our prediction of land
suitability relies on various predictors, including elevation, growing degree days,
and soil moisture surplus. Notably, these
predictors are available at different resolutions than the land cover data (25$m$).


The paper is structured as follows. In Section \ref{sec:slgm}, we introduce spatial latent
Gaussian models as a comprehensive framework encompassing common models used
in spatial statistics. Section
\ref{sec:cos} builds upon this foundation and extends the discussion to the properties of
these processes when projected onto aggregated supports. We present our proposed Bayesian
latent Gaussian spatial model with change of support in this section. To demonstrate the
reliability and flexibility of our approach, we conduct simulation studies in Section
\ref{sec:simulation-study}. In Section \ref{sec:case-studies}, we model the land suitability
of improved grassland in Rhondda Cynon Taf. Finally, in Section \ref{sec:discussion}, we
discuss the contributions, advantages, disadvantages, limitations, and potential future
directions of our work.


\section{Spatial latent Gaussian models}%
\label{sec:slgm}

Spatial latent Gaussian models (SLGM) encompass
commonly used models for analysing spatial point-level and area-level data. They
can be tailored to specific classical generalised spatial models depending on the latent
stochastic process used to capture spatial variation. We define a
SLGM in Section \ref{sub:slgm-definition}
and discuss the limitations concerning their ability to handle
spatial data with different levels of support in Section \ref{sub:slgm-limitations}.

\subsection{Definition}%
\label{sub:slgm-definition}

A classical SLGM assumes that the conditional distribution
$\df{F}(\cdot)$ of a random function $\rv{Y}(\loc)$, for a location $\loc \in D \subset R^2$,
given a $p$-dimensional set of covariates $\ve{x}(\loc) = [x_1(\loc), \dots ,
x_p(\loc)]^\intercal$ and the Gaussian random function $\ssp(\loc)$, can be defined as
follows
\begin{align}
\label{eq:slgm}
\rv{Y}(\loc) \mid \ve{x}(\loc), \ssp(\loc) & \sim \df{F}(\theta(\loc), \tau^2), \nonumber \\
    g(\theta(\loc)) & = \eta(\loc) =
    \beta_0 + \ve{x}^\tr(\loc)\ve{\beta} + \ssp(\loc),
\end{align}
where $g(\cdot)$ is the link function between the conditional mean $\theta(\loc)$ and the
linear predictor $\eta(\loc)$, and $\tau^2$ is an additional parameter. As usual,
\(\beta_0\) and \(\ve{\beta}\) are the intercept and covariate effects, respectively. For
a set of locations $\locset = \{s_1, \dots, s_n\}$, which can be points, lines or regions,
this implies that the elements of a set of random variables \(\{Y(s_i): i = 1, 2, \dots,
n\}\) are conditionally independent given the covariates \(\{\ve{x}(s_i)\}\) and the
spatial stochastic process \(\{\ssp(\loc): \loc \in D\}\). This model assumes that the
random vector $\rve{\ssp}_l = [\ssp(\loc_1), \dots, \ssp(\loc_n)]$ follows a multivariate
Gaussian distribution.

The model defined by the \eq \eqref{eq:slgm} encompasses various classical spatial models
based on the specific type of $\loc$ (point or regions) and the properties of the spatial
stochastic process $\{\ssp(\loc)\}$. There are three main cases: i) In the classical
generalised geostatistical model, $\loc$ represents a point, and $\{\ssp(\loc)\}$ is
defined as a Gaussian process with zero-mean and covariance function
$\kappa(\loc,\loc^*))$. ii) In the generalised conditional autoregressive (CAR) spatial
model, $\loc$ represents a region and $\{\ssp(\loc)\}$ is defined as a Gaussian Markov
random field (\gmrf) with zero-mean and precision matrix $\m{Q}$. iii) In the geoadditive
model, $\loc$ represents a point, and $\{\ssp(\loc)\}$ is defined using basis functions.
Details of common stochastic processes can be found in \se \ref{sub:slgm-spp} of the
Supplementary Material (SM).

\subsection{Limitations}%
\label{sub:slgm-limitations}

Although the framework defined by \eq \eqref{eq:slgm} is attractive and encompasses
different types of spatial models, it has certain requirements. Firstly, it assumes that
responses are observed at the same type of support, such as either point-level or
aggregated-level data. Secondly, both responses and predictors need to be available at the
same sampling units to make inferences. Lastly, predictors should be available for
any location within the area of interest ($\locs \in \mathcal{D}$) to enable spatial
prediction at unobserved locations. However, these assumptions are often challenging to
satisfy in real-world applications and are closely tied to the concept of \textit{change
of support} discussed earlier.

In practice, it is common for practitioners to perform pre-processing steps when using the
models defined in Equation \ref{eq:slgm}. This is because the three assumptions mentioned
earlier are often not met, and existing approaches for dealing with change of support are
often inflexible. These approaches typically involve approximating covariance matrices on
aggregated supports using Monte Carlo integration or by averaging spatial random
variables. A common practice is to interpolate the spatial data to a consistent support,
which enables the application of models like \eqref{eq:slgm}. However, it can introduce
biases in mean predictions and also affect the accuracy of uncertainty
quantification.

Our research on this topic proposes models for continuous spatial variation
able to (i) include data at different support types for the response variable, (ii)
include predictors observed at different spatial sampling units, (iii) perform spatial
prediction of the response variable $Y(\locs)$ even in cases where the covariates
$\ve{x}(\locs)$ are not available for all $\locs \in D$.


\section{\large Spatial latent Gaussian model with change of support}%
\label{sec:cos}

 We present a hierarchical SLGM that addresses the
 challenge of handling spatial data with different supports in the response variable and
 predictors. We introduce the concept of change of support in spatial
 stochastic processes and discuss the properties of Gaussian processes (GPs) and linear
 combinations of basis functions when the support is altered (Section \ref{sub:cos-sp}).
 Subsequently, in Section \ref{sub:cos}, we propose a hierarchical SLGM
 that effectively handles the change of support by utilising latent spatial
 processes.
 Bayesian inference and spatial prediction of our approach are explained in
 sections \ref{sub:cos-inference} and \ref{sub:cos-prediction} of the SM.

\subsection{Change of Support on Stochastic Processes}%
\label{sub:cos-sp}

Let $\{\ssp(\locs): \locs \in \mathcal{D}\}$ represent a spatial stochastic process with a
continuous index set $\mathcal{D}$ and continuous state space $\mathbb{R}$. When the
support is changed, the process is defined over a different index set, resulting in
another process $\{\ssp(\locc): \locc \subset \mathcal{D}\}$, where $\locc$ represents any
geometry included in $\mathcal{D}$. Specifically, we focus on the case where $\locc$
represents a region, and $\ssp(\locc)$ is defined as a integral over that region:
\begin{equation}
  \label{eq:cos}
  \ssp(\locc) = \int_\locc \ssp(\locs)h(\locs) d\locs.
\end{equation}
In this equation, the function $h(\ve{s})$ serves as a weighting function that determines
the importance of each location $\ve{s}$ within the region $\locc$. In practical
applications, this function can take into account factors such as the sampling effort at a
specific location. Alternatively, $h(\ve{s})$ can be defined as $\abs{\locc}^{-1}$ to
obtain an average over $\locc$, or as $1$ for a total.

To effectively integrate datasets with different spatial supports through the change of
support, it is necessary to use stochastic processes $\{\rv{\ssp}(\ve{s})\}$ that allow
for efficient and accurate computation of \eq \eqref{eq:cos}. Specifically, when
considering a set of points $\mathcal{S} = \{\locs_i: i = 1, \dots, n_\locs\}$ and regions
$\mathcal{C} = \{\locc_i: i = 1, \dots, n_\locc\}$, we need the capability to obtain the
joint density of the associated random vectors $\rve{\ssp}_\locs = [\ssp(\locs_1),
\ssp(\locs_2), \dots, \ssp(\locs_m)]^\top$ and $\rve{\ssp}_\locc = [\ssp(\locc_1),
\ssp(\locc_2), \dots, \ssp(\locc_m)]^\top$. Furthermore, we should be able to establish
the connection or association between these random vectors $\rve{\ssp}_\locs$ and
$\rve{\ssp}_\locc$. In the following section, we discuss the properties with respect to
the change of support for Gaussian processes and linear combinations of basis functions.

\subsubsection{Gaussian Processes}%
\label{ssub:cos-sp-gp}

The main property for a $\gp(\mu(\locs), \kappa(\locs, \locs^*))$ is that the random
vector $\rve{\ssp}_\locs$ for any finite set of points $\mathcal{S} = \{\locs_i: i = 1,
\dots, n_\locs\}$ follows a multivariate normal distribution with vector mean
$\ve{\mu}_\locs = [\mu(\locs_1), \dots, \mu(\locs_{n_\locs})]$ and covariance matrix
$\m{\Sigma}_\locs$ with elements $(\m{\Sigma}_\locs)_{ij} = \kappa(\ve{s}_i,\ve{s}_{j})$.
When the support is changed, as presented in \eq \eqref{eq:cos}, the random vector
$\rve{\ssp}_\locc$ for any finite set of regions $\mathcal{C} = \{\locc_i: i = 1, \dots,
n_\locc\}$ follows also a multivariate normal distribution with mean $\ve{\mu}_\locc =
[\int_{\locc_1}\mu(\locs)d\locs, \dots, \int_{\locc_n}\mu(\locs)d\locs]$ and covariance
matrix $\m{\Sigma}_\locc$ \citep{gelfand2001change}. The elements of the covariance matrix are
defined by
  ${(\m{\Sigma}_\locc)}_{ij} = \frac{1}{\abs{\locc_i}\abs{\locc_j}}
  \int_{\locs \in \locc_i} \int_{\locs^* \in \locc_j} \kappa(\locs-\locs^*)
  d\locs d\locs^*$,
where $\abs{\cdot}$ is the area operator.

More generally, considering the set of points and regions $\mathcal{A} = \{\locs_1, \dots,
\locs_{n_\locs}, \locc_1, \dots, \locc_{n_\locc}\}$, the associated random vector
$\rve{\ssp}_a$ also follows a multivariate normal distribution with $\ve{\mu}_a =
[\ve{\mu}_\locs^\tr, \ve{\mu}_\locc^\tr]$ and covariance matrix
  $\m{\Sigma}_a =
  \left[\begin{array}{c|c}
    \m{\Sigma}_\locs      & \m{\Sigma}_{\locs\locc} \\ \hline
    \m{\Sigma}_{\locs\locc}^\tr & \m{\Sigma}_\locc
  \end{array}\right]
  $,
where
  ${(\m{\Sigma}_{\locs\locc})}_{ij} = \frac{1}{\abs{\locc_j}} \int_{s \in \locc_j}
  \kappa(\locs_i-\locs) d\locs$.
Hence, it is possible to derive the joint density function for any set of points and/or
regions, which can be utilised for statistical inference and spatial prediction. However,
both depends of the elements of the covariate matrix, which are commonly approximated
using Monte Carlo integration (see \se \ref{ssub:app-cos-sp-gp} of the SM).

%

\subsubsection{Linear Combinations of Spatial Basis Functions}%
\label{ssub:cos-sp-bf}

The evaluation of the linear combination of spatial basis functions
over any n-dimensional set of points $\df{S}$ can be expressed as
$\rve{\ssp}_\locs = \m{B}_\locs\bm{\delta}$ where $\rve{\delta}$ is a n-dimensional GMRF and the
row $i$ of $\m{B}_\locs$ is the evaluation of the basis function at point $\locs_i$. Under
a change of support,
the continuous process is projected to
\begin{align}
  \label{eq:bf_cos}
  \rv{\ssp}(\locc) & = \sum_{j=1}^{q_1}\sum_{l=1}^{q_2} \delta_{jl}
  \left(\int_{\locs \in \locc} b_{jl}(\locs) d\locs\right) = \ve{b}^\tr(\locc)\rve{\delta},
\end{align}
where $\rve{\delta}$ is the $q$-dimensional GMRF and $\ve{b}(\locc)$ is a vector
containing the integral of the two-dimensional basis functions $b_{jl}(\cdot)$ over
$\locc$. Hence, any random vector associated with a set of regions $\df{C}$ can be expressed
as $\rve{\ssp}_\locc = \m{B}_\locc\bm{\delta}$ where the $i$th row of $\m{B}_\locc$ is the
evaluation of the basis functions at region $\locc_i$.

Notice that, under change of support, our process remains similar and we only need to
update the basis functions as integrals over the region of interest while the main latent
process $\rve{\delta}$ does not need any transformation. This provides a close connection
between a subset of the process over any finite set of points $\df{S}$ and a set of regions
$\df{C}$. Both are simply a linear transformation of the latent GMRF $\rve{\delta}$ with
different and known design matrices $\m{B}_\locs$ and $\m{B}_\locc$ respectively.

\paragraph{Integrating basis functions:}

An important aspect in \eq \eqref{eq:bf_cos} is that we should be able to integrate the basis
functions over any arbitrary region $\locc$. Given that we define two dimensional basis
functions as the product of uni-dimensional basis functions, then the integral is
$
\int_{\locs \in \locc} b_{jl}(\locs) d\locs =
\int_{(s_1, s_2) \in \locc} b^1_{j}(s_1) b^2_{l}(s_2)ds_1ds_2
$.
This expression can be reduced for rectangular regions such as $s_1 \in [L_1, U_1]$ and
$s_2 \in [L_2, U_2]$,
$
  \int_{\locs \in \locc} b_{jl}(\locs) d\locs =
  \int_{L_1}^{U_1} b^1_{j}(s_1)ds_1 \times
  \int_{L_2}^{U_2} b^2_{l}(s_2)ds_2,
$
as it is only required to compute the integral of the univariate basis functions.
We use basis functions because the integral can be computed efficiently due to:
\begin{enumerate}
  \item The j-th basis spline $B_{jk}(x)$ of order $k$ is non-zero from knot $t_j$ to
    $t_{j+k}$.
  \item The integral from knot $t_j$ to an arbitrary value $x$ is
    \begin{equation*}
      \int_{t_j}^x B_{j,k}(t)dt =
      \left\lbrace
      \begin{array}{cl}
        \displaystyle\frac{t_{j+k} - t_j}{k} \sum_{i = j}^{s-1} B_{i,k+1}(x), &  t_j < x \leq t_{j+k}\\
        \displaystyle\frac{t_{j+k} - t_j}{k} \sum_{i = j}^{j+k-1} B_{i,k+1}(x), & t_{j+k} < x \\
      \end{array}
      \right.
    \end{equation*}
    for $s$ such that $t_{s-1} < x \leq t_s$.
  \item The integral can be evaluated from $t_i$ to $x$ such that $i\geq j$ and $t_i < x
    \leq t_{j+k}$,
    \begin{equation*}
      \int_{t_i}^x B_{j,k}(t)dt =
      \left\lbrace
      \begin{array}{cl}
        \displaystyle\frac{t_{j+k} - t_j}{k}
        \left(\sum_{r=0}^{s-1}B_{j+r,k+1}(x) - \sum_{r=0}^{i-j-1}B_{j+r,k+1}(t_i)\right),
        & t_j < x \leq t_{j+k}\\
        \displaystyle\frac{t_{j+k} - t_j}{k}
        \left(\sum_{r=0}^{j+k-1}B_{j+r,k+1}(x) - \sum_{r=0}^{i-j-1}B_{j+r,k+1}(t_i)\right),
        & t_{j+k} < x \\
      \end{array}
      \right.
    \end{equation*}
    for $s$ such that $t_{s-1} < x \leq t_s$.
\end{enumerate}
These results led to efficient computation of the integral by: i) evaluating it only when
required, according to the local definition of basis splines, and ii) computing the
integral depending only on few basis splines at order $k+1$. The proofs of these results
are provided in Section \ref{sec:app-basis-splines} of SM.
Hence, using linear combinations of basis functions we can integrate the process
efficiently and exactly in rectangular geometries.
Details of inference and prediction can be found in \se
\ref{ssub:app-cos-sp-bf} of SM.

\subsection{Spatial Latent Gaussian Model with Change of Support}%
\label{sub:cos}

We propose a model-based approach to integrate spatial data by extending
the spatial latent Gaussian models (SLGM) discussed in Section \ref{sec:slgm} and
leveraging the principles of change of support as outlined in Section \ref{sub:cos-sp}. We
begin by presenting a general framework for SLGM with change of support. Later, we provide
specific models for the Gaussian and Bernoulli cases and describe key details of Bayesian
inference for these models.

\subsubsection{General framework}%
\label{ssub:general-framework}

Our approach is founded on the existence of latent continuous processes that underlie the
fundamental mechanisms of the studied phenomena. It acknowledges that change of support
could also be required under transformations of the latent processes and considers that
measurement error is shaped by the characteristics of the observations, independent of the
latent process. With these fundamental principles in mind, our approach comprises three
key components: the \textit{latent Gaussian model}, the \textit{change of support model},
and the \textit{observation model}. A directed graph representing our model across the
three layers is presented in \fig \ref{fig:dag}.

The \textit{latent Gaussian model} closely resembles a linear geostatistical model, where
a latent response process $\{\rv{\eta}(\locs)\}$ is expressed as a function of zero-mean
latent predictors $\{\rv{V}_j(\locs)\}$ for $j=1,\dots,p$ and a zero-mean latent spatial
process $\{\ssp(\locs)\}$, capturing additional spatial variation. This relationship is
defined by the equation:
\begin{equation}
  \label{eq:lgm}
  \eta(\locs) = \beta_0 + \rve{V}(\locs)^\tr\ve{\beta} + \ssp(\locs), \quad~\text{for}~
  \ve{\locs}
  \in \locsset,
\end{equation}
In this model, $\beta_0$ serves as an intercept parameter, $\mean{\eta(\locs)} = \beta_0$,
while $\ve{\beta}$ represents a vector of regression coefficients associated with the
latent predictors $\rve{V}(\locs)$. The latent processes $\{\eta(\locs)\}$ and
$\{\rv{V}(\locs)\}$ are linked to the responses and predictors, respectively, which can be
observed at either point or aggregated levels.

\begin{figure}[htb]

\begin{tikzpicture}[>=latex,line join=bevel,scale=0.55]
\begin{scope}
  \pgfsetstrokecolor{black}
  \definecolor{strokecol}{rgb}{1.0,1.0,1.0};
  \pgfsetstrokecolor{strokecol}
  \definecolor{fillcol}{rgb}{1.0,1.0,1.0};
  \pgfsetfillcolor{fillcol}
  \filldraw (0.0bp,0.0bp) -- (0.0bp,508.0bp) -- (824.06bp,508.0bp) -- (824.06bp,0.0bp) -- cycle;
\end{scope}
\begin{scope}
  \pgfsetstrokecolor{black}
  \definecolor{strokecol}{rgb}{1.0,1.0,1.0};
  \pgfsetstrokecolor{strokecol}
  \definecolor{fillcol}{rgb}{1.0,1.0,1.0};
  \pgfsetfillcolor{fillcol}
  \filldraw (0.0bp,0.0bp) -- (0.0bp,508.0bp) -- (824.06bp,508.0bp) -- (824.06bp,0.0bp) -- cycle;
\end{scope}
\begin{scope}
  \pgfsetstrokecolor{black}
  \definecolor{strokecol}{rgb}{1.0,1.0,1.0};
  \pgfsetstrokecolor{strokecol}
  \draw (8.0bp,8.0bp) -- (8.0bp,500.0bp) -- (430.18bp,500.0bp) -- (430.18bp,8.0bp) -- cycle;
\end{scope}
\begin{scope}
  \pgfsetstrokecolor{black}
  \definecolor{strokecol}{rgb}{1.0,1.0,1.0};
  \pgfsetstrokecolor{strokecol}
  \draw (8.0bp,8.0bp) -- (8.0bp,500.0bp) -- (430.18bp,500.0bp) -- (430.18bp,8.0bp) -- cycle;
\end{scope}
  \node (x1) at (63.5bp,419.0bp) [draw,rectangle] {$\bm{x}_{1}$};
  \node (x2) at (63.5bp,254.0bp) [draw,rectangle] {$\bm{x}_{2}$};
  \node (ws1) at (220.44bp,419.0bp) [draw,circle,fill=blue!20] {$V_1(l^{[x]}_1)$};
  \node (y) at (493.28bp,254.0bp) [draw,circle,fill=green!20] {$\eta(\ve{s})$};
  \node (ws3) at (220.44bp,89.0bp) [draw,circle,fill=blue!20] {$V_3(l^{[x]}_3)$};
  \node (ws2) at (220.44bp,254.0bp) [draw,circle,fill=blue!20] {$V_2(l^{[x]}_2)$};
  \node (w4) at (376.03bp,34.0bp) [draw,circle,fill=green!20] {$W(\ve{s})$};
  \node (w3) at (376.03bp,144.0bp) [draw,circle,fill=green!20] {$V_3(\ve{s})$};
  \node (w2) at (376.03bp,254.0bp) [draw,circle,fill=green!20] {$V_2(\ve{s})$};
  \node (w1) at (376.03bp,364.0bp) [draw,circle,fill=green!20] {$V_1(\ve{s})$};
  \node (x3) at (63.5bp,89.0bp) [draw,rectangle] {$\bm{x}_{3}$};
  \node (y1) at (783.06bp,334.0bp) [draw,rectangle] {$\bm{y}_1$};
  \node (ys2) at (635.22bp,174.0bp) [draw,circle,fill=blue!20] {$\eta(l_2^{[y]})$};
  \node (ys1) at (635.22bp,334.0bp) [draw,circle,fill=blue!20] {$\eta(l^{[y]}_1)$};
  \node (y2) at (783.06bp,174.0bp) [draw,rectangle] {$\bm{y}_2$};
  \draw [->] (x1) ..controls (119.12bp,419.0bp) and (127.75bp,419.0bp)  .. (ws1);
  \draw [->] (x2) ..controls (119.12bp,254.0bp) and (127.75bp,254.0bp)  .. (ws2);
  \draw [->] (w1) ..controls (425.75bp,317.35bp) and (444.5bp,299.77bp)  .. (y);
  \draw [->] (w3) ..controls (425.75bp,190.65bp) and (444.5bp,208.23bp)  .. (y);
  \draw [->] (w4) -- (y);
  \draw [->] (x3) ..controls (119.12bp,89.0bp) and (127.75bp,89.0bp)  .. (ws3);
  \draw [->] (w1) ..controls (322.11bp,383.06bp) and (310.84bp,387.05bp)  .. (ws1);
  \draw [->] (y2) ..controls (733.87bp,174.0bp) and (725.1bp,174.0bp)  .. (ys2);
  \draw [->] (y1) ..controls (733.87bp,334.0bp) and (725.1bp,334.0bp)  .. (ys1);
  \draw [->] (y) ..controls (536.03bp,278.09bp) and (550.24bp,286.11bp)  .. (ys1);
  \draw [->] (y) ..controls (536.03bp,229.91bp) and (550.24bp,221.89bp)  .. (ys2);
  \draw [->] (w3) ..controls (322.11bp,124.94bp) and (310.84bp,120.95bp)  .. (ws3);
  \draw [->] (w2) ..controls (430.83bp,254.0bp) and (439.66bp,254.0bp)  .. (y);
  \draw [->] (w2) ..controls (321.61bp,254.0bp) and (313.03bp,254.0bp)  .. (ws2);
  \draw [dashed] (10.5, -0.5) rectangle (19.5, 15.5);
  \node[text width=5cm, anchor=west] at (10.5,14.85) {\footnotesize(I) \textit{Latent Gaussian model}};
  \draw [dashed] (5, -0.75) rectangle (25, 17.75);
  \node[text width=8cm, anchor=west] at (5,17) {\footnotesize(II) \textit{Change of support model}};
  \draw [dashed] (1, -1) rectangle (29, 19.25);
  \node[text width=8cm, anchor=west] at (1,18.5) {\footnotesize(III) \textit{Observation model}};
\end{tikzpicture}
  \caption{
    A spatial latent Gaussian model with a change of support, involving three observed
    predictors ($\bm{x}_i$) and two response datasets ($\bm{y}_j$). $\{\eta(\locs)\}$
    represents the latent dependent stochastic process, while $\{V_i(\locs)\}$ are latent
    predictors, and $\{W(\locs)\}$ is a latent process accounting for unexplained spatial
    variability. $\{V_i(l_i^{[x]})\}$ depicts the latent predictors projected to the
    supports of the observed predictors, while $\{\eta(l_j^{[y]})\}$ shows the projection
    of the latent dependent process to the supports of the observed responses.
  }
  \label{fig:dag}
\end{figure}
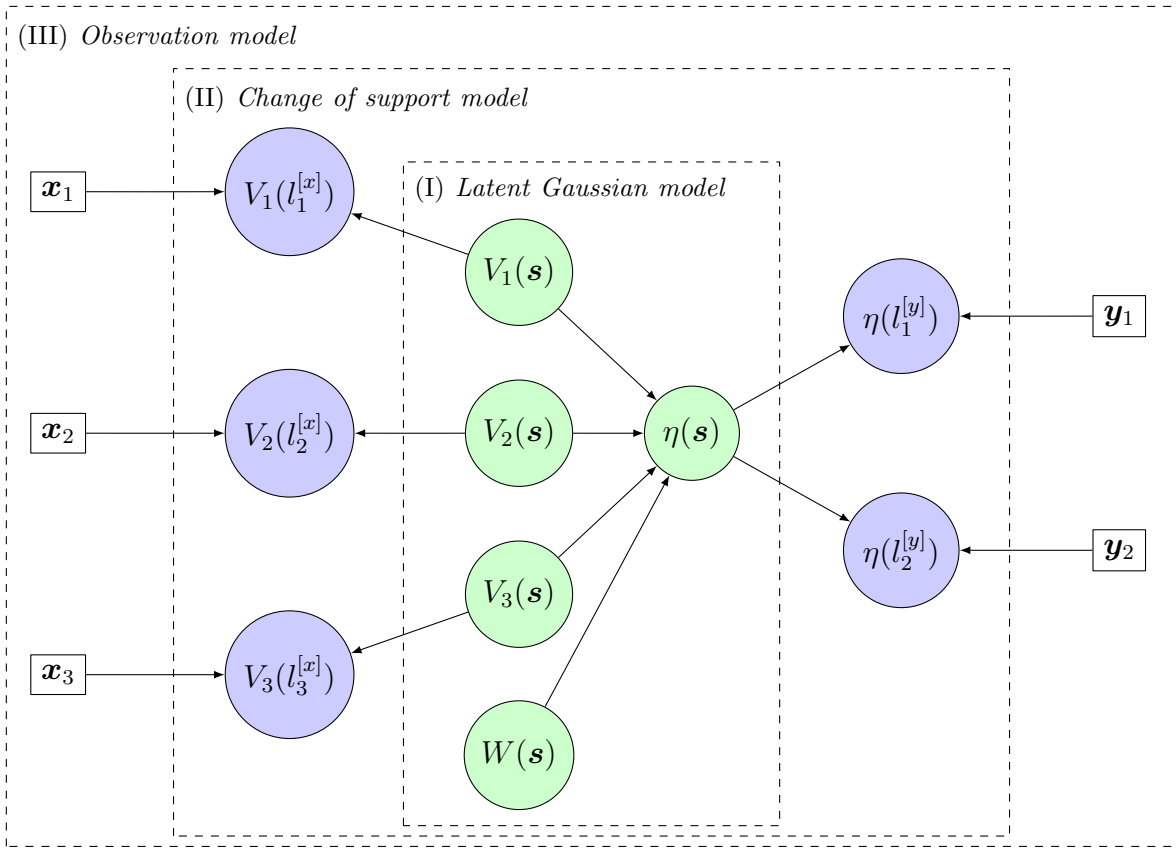

The \textit{change of support model} elucidates how processes within the latent Gaussian model can be projected onto different supports. This projection is not always required for
the processes themselves but may be necessary for transformations. We define the change of
support for the processes $\{\ssp(\locs)\}, \{V_j(\locs)\}$ and $\{\eta(\locs)\}$ under
transformations $h_w(\cdot), h_{v_j}(\cdot)$ and $h_{\eta}(\cdot)$ over the geometry $\locc
\subset \locsset$ as:
\begin{align*}
  \ssp(\locc, h_{w}) & = \frac{\int_{\locs \in \locc} h_w(\ssp(\locs)) d\locs}{\abs{\locc}},
  &
  \rv{V}_j(\locc, h_{v_j}) & = \frac{\int_{\locs \in \locc} h_{v_j}(\rv{V}_j(\locs)) d\locs}{\abs{\locc}},
  &
  \rv{\eta}(\locc, h_{\eta}) & = \frac{\int_{\locs \in \locc} h_{\eta}(\rv{\eta}(\locs))
  d\locs}{\abs{\locc}}.
\end{align*}
This notation emphasises that the change of support under a transformation is distinct from
the transformation of the change of support, i.e., \(\ssp(\locc, h_{w}) \neq
h_{w}(\ssp(\locc))\). The choice of the transformation \(h(\cdot)\) depends on the
relationship between the models proposed at the point and aggregated levels.

The \textit{observation model} defines the distribution or data generation mechanism
for observable variables (responses or predictors) at any location (point or geometry).
This model incorporates the change of support processes and additional parameters.
Specifically, let $y_{ki}$ represent the observed response value at location
$\loc_{ki}^{[y]}$ for the $i$-th sampling unit of the $k$-th source of information, where
$k=1,\dots,K$ and $i=1,\dots,n_k$. The associated random variable follows:
$
\rv{Y}_{k}(\loc_{ki}^{[y]}) \sim \mathcal{F}(\rv{\eta}(\loc_{ki}^{[y]}, h_{\eta}), \ve{\alpha}_{k}),
$
where $\ve{\alpha}_k$ denotes additional parameters specific to the source of information
$k$, required to define the data generation mechanism $F(\cdot,\cdot)$. These parameters
can account for measurement errors or mean biases across different data sources. Notice
that when observations are at point level, then $\rv{\eta}(\loc_{ki}^{[y]}, h_{\eta}) =
h_{\eta}(\rv{\eta}(\loc_{ki}^{[y]}))$.

In a similar fashion, for observed predictor values at location $\loc_{ji}^{[x]}$ for the
$i$-th sampling unit of the $j$-th predictor, with $j=1,\dots,p$ and $i=1,\dots,m_j$, the
random variable associated with these observed predictors follows:
$
\rv{X}_{j}(\loc_{ji}^{[x]}) \sim \mathcal{F}(\rv{\eta}(\loc_{ji}^{[x]}, h_{\eta}), \ve{\nu}_{j}),
$
where $\ve{\nu}_j$ also represent additional parameters, primarily aimed at defining
the variability of the measurement error in the predictors.

We have refrained from imposing specific distributions or data generation mechanisms, as
these depend on the particular random variables and point-level models. In the following
sections, we will provide specific details when dealing with Gaussian and
Bernoulli-distributed response random variables.

\subsubsection{Gaussian case}%
\label{ssub:gaussian-case}

In the Gaussian case, both the responses and predictors are assumed to follow a normal
distribution. The \textit{latent Gaussian model} for the process $\{\eta(\locs)\}$ remains
consistent with \eq \ref{eq:lgm}. We employ this model with the rationale that the latent
process $\{\eta(\locs)\}$ is directly linked to the response of interest
$\{\rv{Y}(\locs)\}$ in such a way that, in the absence of measurement error,
$\rv{Y}(\locs) = \rv{\eta}(\locs)$. Consequently, the aggregated model for any location
$\locc \subset \locs$ could be simply defined as $\rv{Y}(\locc) = \abs{c}^{-1}\int_{\locs
\in \locc}\rv{\eta}(\locs)d\locs$. Because that change of support is required directly on
the latent process (i.e. an identity transformation), the \textit{support model} with
respect to an arbitrary region $\locc \subset \locsset$ comprehend the following:
\begin{align}
  \label{eq:gaussian-model-cos}
  \ssp(\locc) & =  \abs{\locc}^{-1}\int_{\ve{s} \in \locc} \ssp(\ve{s}) d\ve{s}, \nonumber\\
  \rv{V}_j(\locc) & = \abs{\locc}^{-1} \int_{\ve{s} \in \locc} \rv{V}_j(\ve{s})d\ve{s},
  ~\text{for}~ j=1,\dots,p, \nonumber \\
  \rv{\eta}(\locc) & = \beta_0 + \rve{V}^\tr(\locc)\ve{\beta}_1 + \ssp(\locc) .
\end{align}
Notice that, using spatially weighted spline functions, the integrals of the processes are
reduced to the linear combinations of the integral with respect to the basis functions as
presented in \se \ref{sub:cos-sp}.

To define the \textit{observation model}, it is important to note that $\{\eta(\locs)\}$ and
$\{\eta(\locc)\}$ represent the process of interest at the point and aggregated levels
without measurement error. Therefore, the response variables at the point and aggregated
levels are simply a noisy version of the latent process $\{\eta(\locs)\}$. The extent of
measurement error depends on the scale of the sampling units and the measurement
instruments associated with the sources of information. Thus, we assume that the
characteristics of measurement errors are independent among different sources. Hence,
considering $\loc_{ki}^{[y]}$ as the location (point or geometry) of the $i$-th
sampling unit in the $k$-th source of information for the response variable, the response
model accounting for measurement error can be written as:
\begin{align*}
  \rv{Y}_{k}(\loc^{[y]}_{ki}) & = b_k + \rv{\eta}(\loc^{[y]}_{ki}) + \varepsilon_k(\loc^{[y]}_{ki}),
  \quad~\text{for}~ k = 1,\dots,K ~\text{and}~ i = 1, \dots, n_k,
\end{align*}
where $b_k$ and $\varepsilon_k(\loc_{ki}^{[y]}) \sim N(0, \tau^2_k(\loc_{ki}^{[y]}))$
are the intercept and error term of the
$k$-th source of information, respectively. To ensure identifiability of the model, we set
$b_k$ to zero for the most reliable data source, and the remaining $b_k$ terms are
interpreted as mean biases relative to the reliable data source. The error term is assumed
to be zero-mean and normally distributed with a variance function $\tau^2_k(\loc)$. If the
sampling units for a source of information are of the same size, then the variance
function can be defined as a constant (i.e., $\tau^2_k(\loc) = \tau^2_k$). However, if the
sampling units are irregular in size, the variance function can be defined in terms of the
size of the sampling units (e.g., $\tau^2_k(\loc) = \tau^2_k\log(\abs{\loc})$).

Likewise, considering $\loc_{ji}^{[x]}$ as the location of the $i$-th
sampling unit for the $j$-th predictor, the predictor model accounting for measurement
error can be expressed as:
\begin{align}
  \label{eq:gaussian-model-predictor}
  \rv{X}_j(\loc_{ki}^{[x]}) & = \alpha_j + \rv{V}_j(\loc_{ji}^{[x]}) +
  \xi_j(\loc_{ji}^{[x]}),
  \quad~\text{for}~ j = 1,\dots,p ~\text{and}~ i = 1, \dots, m_j,
\end{align}
where $\alpha_j$ and $\xi_j(\loc_{ji}^{[x]}) \sim N(0, \psi^2_j(\loc_{ji}^{[x]}))$
represent the intercept and error term of
the $j$-th predictor, respectively. The variance function $\psi^2_j(\loc)$ of the error
term can also be defined with respect to the observed sampling units, as explained above.

\subsubsection{Bernoulli case}%
\label{ssub:bernoulli-case}

Our approach can be used for modelling binary outcomes with spatial structure under a change
of support. We initially describe the underlying relationship between the point and
aggregate level models without measurement error before presenting the complete model.

Let $\{\rv{Y}(\locs): \locs \in \locsset\}$ represent a spatial process obtained by
binarising the continuous latent process $\{\eta(\locs)\}$ with respect to an unknown
threshold $\gamma$, as follows:
\begin{align*}
  \rv{Y}(\locs) & = \left\lbrace\begin{array}{cl}
    1, & \eta(\locs) > \gamma, \\
    0, & \eta(\locs) \leq \gamma.
  \end{array}\right.
\end{align*}
Here, $\{\eta(\locs)\}$ is the latent process defined in \eq \eqref{eq:lgm}. Under his
model, the probability of success at location $\locs$ is given by
$\Phi((\beta_0-\gamma)/\sqrt{\var{\rv{\eta}(\locs)}})$, where $\Phi(\cdot)$ represents the
cumulative distribution function of a standard normal distribution. This model is not
identifiable due to two reasons: (i) an increase in the threshold $\gamma$ can be offset
by an increase in the intercept $\beta_0$, and (ii) a multiplicative factor in
$\ssp(\locs)$ can be compensated by a multiplicative factor in $\gamma$, $\beta_0$, and
$\ve{\beta}$. An identifiable model can be achieved by fixing a specific value for
$\gamma$ or $\beta_0$, and setting a constant variance for either $\ssp(\locs)$ or
$\varepsilon(\locs)$. To define properly our point-level model, we set $\gamma=0$ and
$\var{\ssp(\locs)}=1$. It is worth noting that defining the threshold as $\gamma=0$ is
equivalent to excluding $\beta_0$ from the latent process $\{\rv{\eta}(\locs)\}$ and using
a threshold of $-\beta_0$.

Since the latent process $\{\rv{\eta}(\locs)\}$ is inherently linked to the response
binary process $\{\rv{Y}(\locs)\}$, it is natural to the define the aggregated spatial
binary process $\{\rv{Y}(\locc): \locc \subset \locsset\}$ in terms of the aggregated
latent process $\{\rv{\eta}(\locc): \locc \subset \locsset\}$. This definition is as
follows:
$
\rv{Y}(\locc) = \left\lbrace\begin{array}{cl}
  1, & \rv{\eta}(\locc) > \gamma, \\
  0, & \rv{\eta}(\locc) \leq \gamma.
\end{array}\right.
$
Here, the binarisation of both point and aggregated levels depends on the same threshold
$\gamma$. Consequently, the probability of success at location $\locc$ is determined by
$\Phi((\beta_0 - \gamma)/\sqrt{\var{\rv{\eta}(\locc)}})$. Because the aggregated model is
derived from the point-level model, the assumptions are propagated. This results in
both $\gamma = 0$, and a fixed value for $\var{\ssp(\locc)}$ due to $\var{\ssp(\locs)} =
1$.

We use the previous mapping between the point-level and aggregated-level models to
formulate the SLGM with a change of support for binary data. First, the \textit{latent
Gaussian model} is the same as in the general case (\eq \ref{eq:lgm}), with the
restriction that $\var{\ssp(\locs)} = 1$. Second, the \textit{change of support model} is
the same as in the Gaussian case (\eq \ref{eq:gaussian-model-cos}) with fixed variance for
$\ssp(\locc)$ due to $\var{\ssp(\locs)} = 1$.

To define the \textit{observation model}, consider $\loc_{ki}^{[y]}$ as the
location (point or geometry) of the $i$-th sampling unit in the $k$-th source of
information for the response variable. The response model can be expressed as:
\begin{align*}
  \rv{Y}_k(\loc_{ki}^{[y]}) & = \left\lbrace\begin{array}{cl}
    1, & \rv{Z}_k(\loc_{ki}^{[y]}) > 0, \\
    0, & \rv{Z}_k(\loc_{ki}^{[y]}) \leq 0.
  \end{array}\right.
  \quad~\text{for}~ k = 1,\dots,K ~\text{and}~ i = 1, \dots, n_k, \\
   \rv{Z}_k(\loc_{ki}^{[y]}) & = b_k + \eta(\loc_{ki}^{[y]}) +
   \varepsilon_k(\loc_{ki}^{[y]}),
\end{align*}
where $\{\rv{Z}_k(\cdot)\}$ is an auxiliary process, for the $k$-th source, that allows us
to introduce bias-related parameters $b_k$ and error terms $\varepsilon_k(\cdot)$ to
account for measurement error. The inclusion of the additional intercept term $b_k$ is
equivalent to reducing the binarising threshold to $\gamma-b_k$. Both approaches can
handle varying the likelihood of success among sources. Similar to the Gaussian case, $b_k$ is
set to zero to the most reliable source of information. The error term
$\varepsilon_k(\loc_{ki}^{[y]}) \sim N(0, \tau^2_k(\loc_{ki}^{[y]}))$
for the $k$-th source is assumed to be zero-mean normally
distributed with a variance function $\tau^2_k(\loc)$, which may depend on the measure of
geometry $\loc$. Finally, the predictor model that accounts for measurement error is
defined as in \eq \eqref{eq:gaussian-model-predictor}.


\section{Simulation Study}
\label{sec:simulation-study}

We performed both one-dimensional (Section \ref{sub:one_dimensional}, SM) and
two-dimensional (Section \ref{sub:two_dimensional}) simulation studies that involve
various configurations of sampling units. Additionally, a simulation study for multiple
binary data and aggregated predictors is presented in Section \ref{sub:sim_full}.

In Section \ref{sub:one_dimensional} and \ref{sub:two_dimensional}, we explore scenarios
where data is observed in the following ways:  \textit{regular grids}, characterised by
uniform sampling unit sizes and regular spacing across the entire area of interest;
\textit{irregular grids}, where sampling unit sizes vary, and spacing is non-uniform but
coverage remains complete; \textit{sparse sampling units}, featuring varying sampling unit
sizes and sparse spatial distributions; and \textit{overlapping sampling units}, where unit
sizes differ, overlap, and coverage is incomplete. Our aim is to underscore the importance
of properly considering spatial data support, showcase the adaptability of our approach to
diverse scenarios, and account for measurement errors.

In Section \ref{sub:sim_full}, we simulate binary data with two sources of information and
predictors, where all the sources have different spatial supports. We demonstrate that
Bayesian inference and prediction are feasible using the methodology described above.

\subsection{Two dimensional}%
\label{sub:two_dimensional}

Considering a set of sampling units $\bm{c} = \{c_1, c_2, \dots, c_n\}$ for any of the
four configurations described above, the corresponding observations are treated as
realisations of the process $\{W(s)\}$ over the sampling units $c_i$ with an additional
measurement error $\varepsilon_i$:
\begin{equation}
  Y_i = \int_{s \in c_i} W(s)ds + \varepsilon_i, \quad \varepsilon_i \sim N(0, \sigma^2_i).
  \label{eq:sim_model}
\end{equation}
We assume that the variance of the measurement error is constant ($\sigma^2_i = \sigma^2$)
when the sampling units have the same size (regular grids). However, it becomes inversely
proportional to the area ($\sigma^2_i = \sigma^2/\abs{c_i}$) when they do not (irregular
grids, sparse sampling units, and overlapping sampling units). It is worth noting that
under this assumption, the measurement error is also influenced by the sampling units.

We define the latent process ${W(s)}$ using a linear combination of basis functions, as
explained in Section \ref{ssub:cos-sp-bf}. When sampling units are arranged in a regular
grid, we compare two models: a \textbf{naive model} (M1) that uses the centroids of the
regions $\locc_i^*$ to model the data ($Y_i = W(\locc_i^*) + \varepsilon_i$) and a
\textbf{support model} (M2) that considers the spatial sampling units, as shown in \eq
\ref{eq:sim_model}. On the other hand, when dealing with irregular grids, sparse sampling
units, and overlapping sampling units, we compare three models: the \textbf{naive model}
(M1) as explained above, a \textbf{heteroscedastic model} (M2) that accounts for the dependence of measurement error on unit sizes, and a \textbf{support and heteroscedastic
model} (M3) that simultaneously considers spatial support and measurement error
heteroscedasticity.

Bayesian inference is carried out through Gibbs sampling. In all of our experiments, we
generated 10,000 MCMC samples, removed the first 2,000 samples, and retained every 1 out
of 5 samples. Subsequently, the continuous latent process is predicted using the MCMC
samples and the basis functions associated with the sampling units. The models are compared
by analysing the characteristics of the predicted mean function, the uncertainty
associated with the prediction, and their posterior probability of overpredicting the true
process at $m$ locations within the study area, defined as:
\begin{equation}
  p(s_j) = \pr{W(s_j) > w(s_j) \mid \ve{y}}, \quad  \text{for}~ j = 1,\ldots,m.
\end{equation}

\subsubsection{Regular grid}
\label{ssub:2d_regular_grid}

For this experiment, we define the latent process ${W(\locs)}$ using 400
basis functions of degree 2 and a GMRF of order 1 with a scale parameter of $\kappa =
0.09$. Similar to the 1-dimensional case, the support model (M2) outperforms the naive
model (M1). Specifically, when examining the predicted mean of M1, it becomes evident that
this model tends to underestimate the peaks and overestimate the valleys, resulting in
more rapid changes (\fig \ref{fig:sim2d_regular}, panels B1-B2). When analysing the
posterior probability of overprediction, we observe that M1 exhibits more frequent values
close to 1 and 0, indicating a tendency for overprediction and underprediction
respectively (\fig \ref{fig:sim2d_regular}-C1). Conversely, the support model
demonstrates a higher frequency of probabilities of overprediction close to 0.5 (\fig
\ref{fig:sim2d_regular}-C2). This behaviour is also linked to the fact that the
uncertainty tends to be underestimated for locations near the centroids of the sampling
units in the naive model.

\begin{figure}[H]
\begin{center}
  \includegraphics[width=\linewidth]{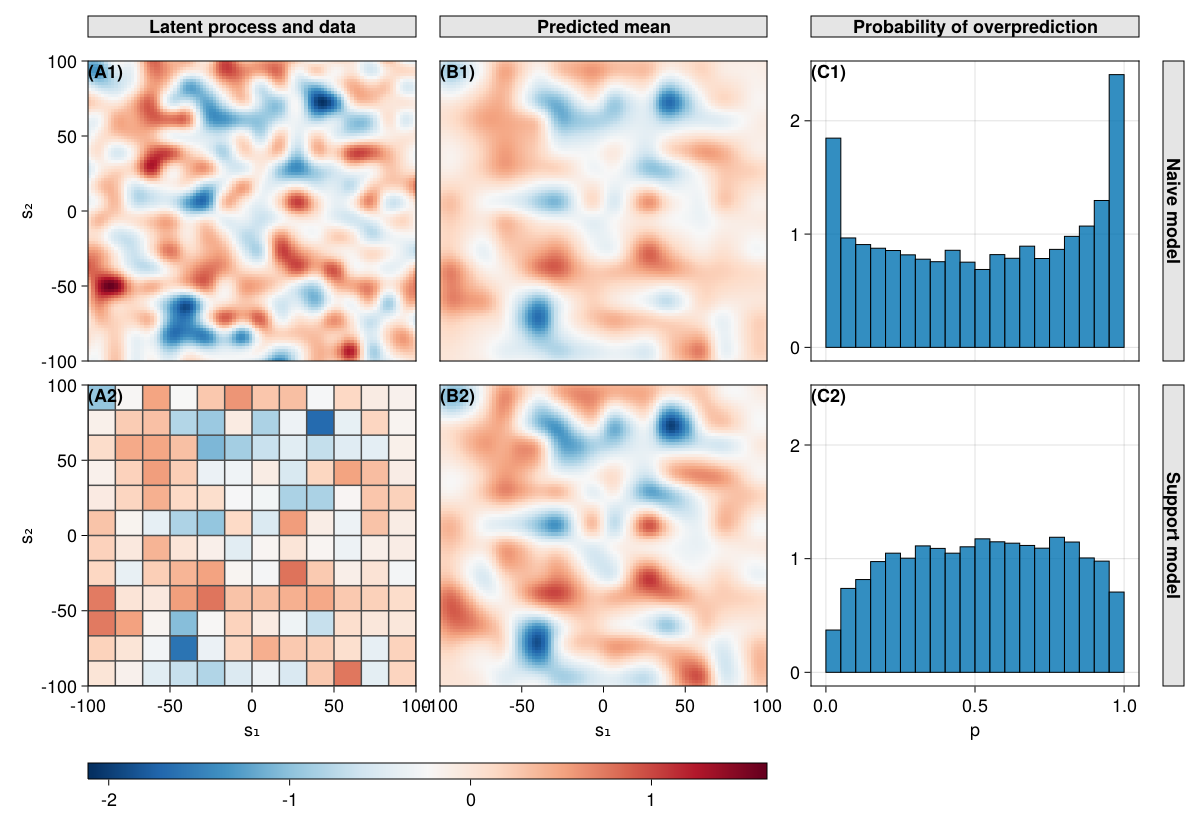}
\end{center}
  \caption{Comparison between the \textbf{naive model} (M1) and the \textbf{support model}
  (M2) when data is observed in a \textbf{regular grid}. Panel A1 displays
  the process of interest, while panel A2 shows the observed data. Panels
  B1-B2 present the predicted mean, and panels C1-C2 feature histograms of
  the posterior probability of overpredicting the underlying process.
  }
  \label{fig:sim2d_regular}
\end{figure}

\subsubsection{Irregular grid, sparse sampling units and overlapping sampling units}
\label{ssub:2d_irregular_grid}

For the irregular grid and sparse sampling units experiments, the latent process
${W(\locs)}$ was defined with 400 basis functions, while 225 basis functions were used for
the experiments with overlapping sampling units. In all three experiments, GMRFs of order
1 with a scale parameter of $\kappa = 0.09$ were employed.

For all three experiments, we consistently observed that both the naive model (M1) and the
heteroscedastic model (M2) tend to underestimate the mean function near peaks and
overestimate it around valleys, a behaviour reminiscent of what was observed in regular
grids (\fig \ref{fig:sim2d_irregular}, \ref{fig:sim2d_sparse},
\ref{fig:sim2d_overlapping} in SM; panels B1-B2). Conversely, the support and heteroscedastic
model (M3) consistently yields predicted means closer to the true process values for both
peaks and valleys (\fig \ref{fig:sim2d_irregular}, \ref{fig:sim2d_sparse},
\ref{fig:sim2d_overlapping} in SM; panels B3). As a result, the predicted mean range for the
support and heteroscedastic model (M3) is generally wider compared to that of M1 and M2.
This underestimation and overestimation tendency observed in M1 and M2 is reflected in the
posterior probability of overprediction, which shows high-frequency values close to 0 and
1 (\fig \ref{fig:sim2d_irregular}, \ref{fig:sim2d_sparse}, \ref{fig:sim2d_overlapping} in
SM;
panels C1-C2). In contrast, the support and heteroscedastic model (M3) exhibits
significantly lower levels of overprediction and underprediction when compared to M1 and
M2 (\fig \ref{fig:sim2d_irregular}, \ref{fig:sim2d_sparse}, \ref{fig:sim2d_overlapping}
in SM;
panel C3).

A noteworthy attribute of the support and heteroscedastic model (M3) is its
proficiency in modelling data acquired under various sampling unit configurations. 
It captures essential properties of the underlying process and provides a suitable
uncertainty quantification, even with sparse and overlapping
sampling units. This underscores the effectiveness of the model in predicting the continuous
latent process when data are observed on aggregated units.

\subsection{Multiple binary data and aggregated predictors}%
\label{sub:sim_full}

In this section, we present a more realistic scenario involving binary response data
observed by two instruments, where the second includes a bias. Predictors are observed at
two different aggregated resolutions. The objective is to recover the parameters of the
system and perform predictions of the latent processes.

Panels (A) and (B) in \fig \ref{fig:sim_full} depict the latent processes
$\rv{V}_1(\locs)$ and $\rv{V}_2(\locs)$ along with observations at regular sampling units
with lengths of $5.71$ and $2.22$, respectively. Panel (C) presents the latent process
$\eta(\locs) = 0.7\rv{V}_1(\locs) - 0.6\rv{V}_1(\locs) + \rv{W}(\locs)$, with noisy values
having unit lengths of $4$ and $2.86$, respectively. It is important to note that the
second source includes a bias of $0.3$. Finally, Panel (D) shows the binary observations
from the two sources at resolutions $4$ and $2.86$, respectively. These observations are
simply a binarisation around $0$ of the noisy values shown in Panel (C).

We conducted Bayesian inference with the aggregated predictors and binary aggregated
outcomes for the change of support model, as outlined in Section
\ref{ssub:bernoulli-case}, using the Gibbs sampling algorithm explained in Section
\ref{ssub:inference}. The algorithm was executed with 10,000 iterations, discarding the
initial 2,000 iterations and retaining every 10th iteration. In Panels (A), (B), and (D)
of Figure \ref{fig:sim_full}, the predictive samples of the latent processes are shown as
grey solid lines.

For the first predictor (Panel A), the predictive samples effectively capture the pattern
of the latent process, and the uncertainty bounds appropriately cover the latent process.
Regarding the second predictor (Panel B), the predictive samples also capture the patterns
of the latent process, exhibit a realistic uncertainty, and show a small mean bias. This
bias is being accounted for by the intercept parameter $\alpha_2$. The key results are
observed in Panel (D), where the predictive samples of the latent process $\eta(\locs)$
accurately capture its patterns. Importantly, these predictions from binary aggregated
outcomes and aggregated predictors do not exhibit the problems of approaches that disregard the support discussed in previous sections. Additionally, the mean bias in the latent
predictor $\rv{V}_2(\locs)$ does not introduce bias in the prediction of the main process
of interest $\eta(\locs)$, as it is effectively compensated by the intercept parameter in
the response variable.

\begin{figure}[H]
\begin{center}
  \includegraphics[width=\linewidth]{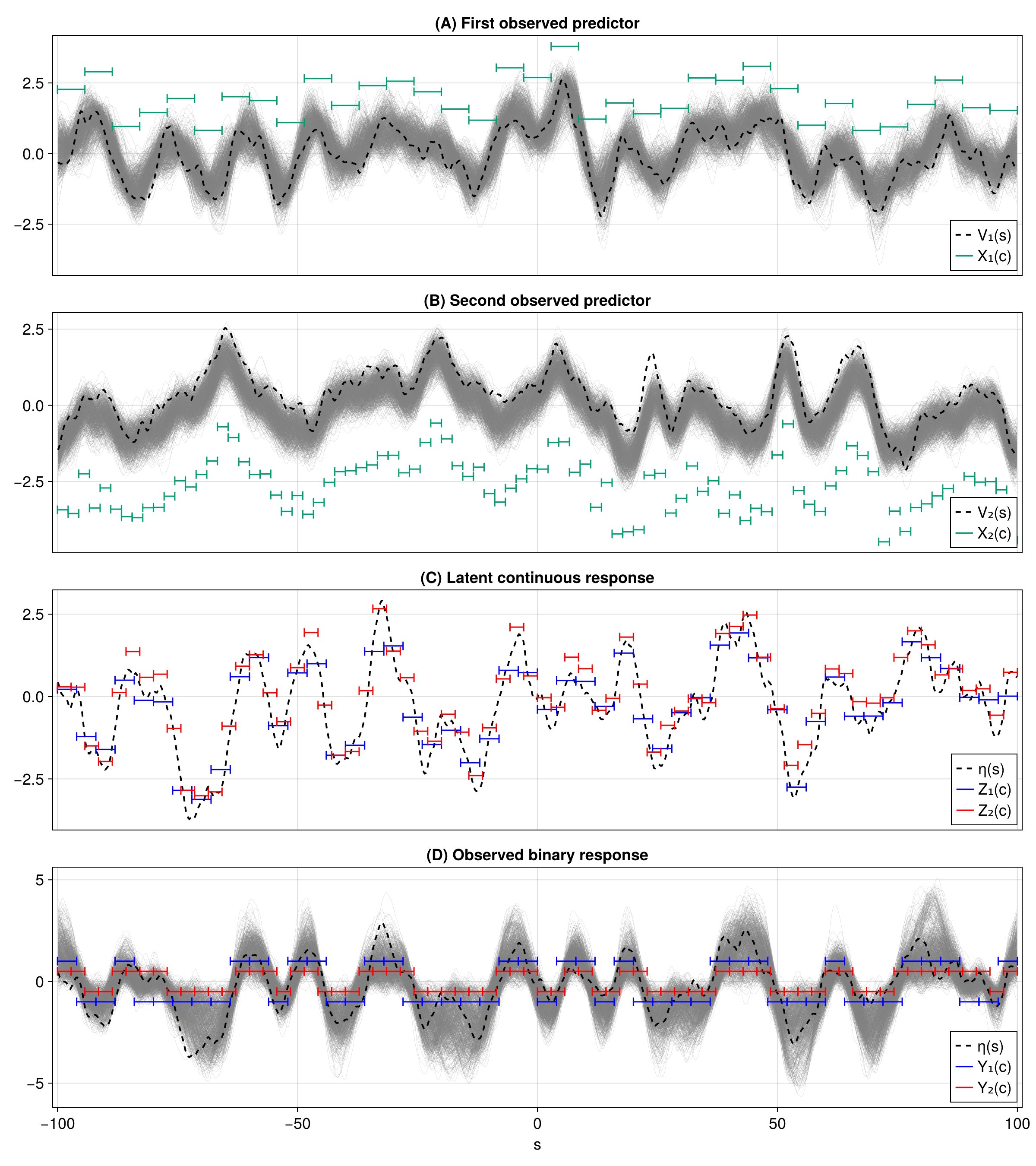}
\end{center}
  \caption{Scenario with two predictors and two binary data sources. Panels (A-B) depict
  latent processes with observed values at aggregated supports. Panel (C) illustrates the
  latent process $\eta(\locs)$ with noisy change of support, and Panel (D) presents binary
  responses. Gray solid lines in all panels represent posterior predictive samples of the
  latent processes.}
  \label{fig:sim_full}
\end{figure}


\section{Land suitability modelling in Rhondda Cynon Taf}%
\label{sec:case-studies}

Land suitability refers to the appropriateness of a piece of land for a specific use or
purpose \citep{jafari2010land}. This concept is crucial in land use planning, resource
management, and environmental conservation. While land suitability itself is not directly
observable, it can be inferred through the examination of climatic, internal soil and
external soil characteristics \citep{wang1994use}. In this section, we employ land cover
data and predictors at various resolutions to predict land suitability for
\textit{improved grassland} in Rhondda Cynon Taf, a county borough in Wales. The resulting
predictions can be used for further analysis, as demonstrated, for example, in
\citet{brown2022agent}, where land suitability surfaces are used as input to generate
future scenarios in the British land use system.

\subsection{Data description}%
\label{sub:ls-data}

The data employed in our land suitability modelling is depicted in Figure
\ref{fig:ls-data}. In Panel (A), we showcase the \textit{land cover} data extracted from
the 2017 UK land cover map at a $25m$ resolution, comprising $884,196$ cells for our
selected study area \citep{marstonland}. This dataset is provided by the UK Centre for
Ecology and Hydrology (UKCEH) and can be accessed at
\url{https://www.ceh.ac.uk/data/ukceh-land-cover-maps}. Moving to Panel (B), we present
\textit{elevation} data at a $25m$ resolution comprising $884,642$ cells, supplied by
Copernicus and available at
\url{https://spacedata.copernicus.eu/collections/copernicus-digital-elevation-model}. It
is important to note that, despite both datasets having a resolution of $25m$, their
support does not align.

Panels (C-D) in Figure \ref{fig:ls-data} depict growing degree days (GDD) and soil
moisture surplus (SMS) at $1km$ resolution comprising $549$ cells. These variables play a
significant role in determining land suitability and are available as annual means for a
period of 20 years (1991-2011) describing the general characteristics of the land. While
air temperature variables (maximum, minimum, and mean) and soil moisture deficit (SMD) are
also available, the former ones exhibit high correlation with growing degree days ($\rho >
0.95$), and the latter is zero constant across the extent of Rhondda Cynon Taf;
consequently, they were not included in our analysis. For detailed information on the
computation of growing degree days refer to \citet{robinson2017trends,
robinson2023climate}. Soil moisture was computed using the method outlined by
\citet{cosby1984statistical}, incorporating evapotranspiration, precipitation, and
available water content \citep{robinson2023climateb, fao2012harmonized}.

\begin{figure}[htb]
\begin{center}
  \includegraphics[width=\linewidth]{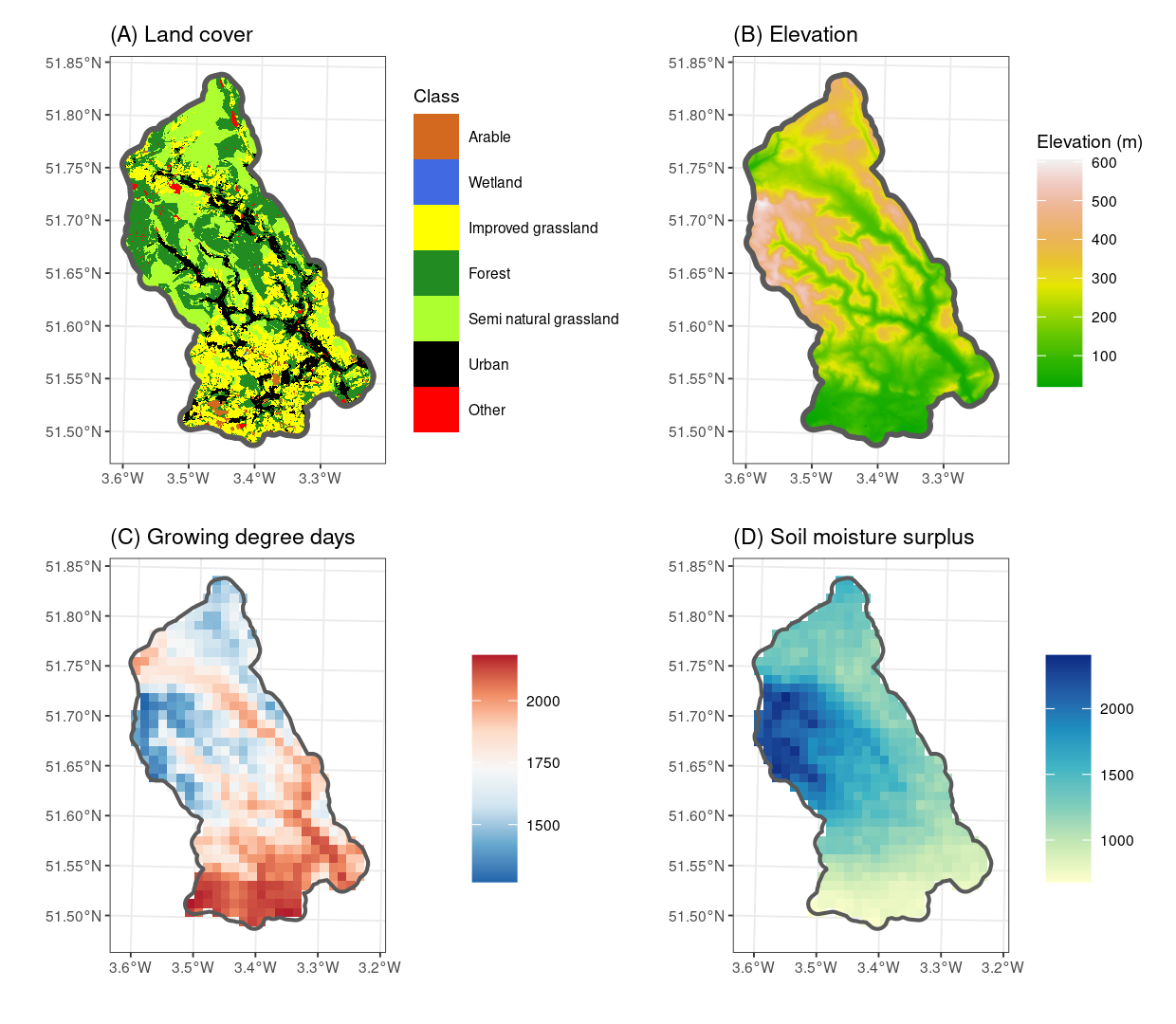}
\end{center}
  \caption{Environmental variables of Rhondda Cynon Taf: (A-B) Land cover and elevation at
  $25$m, and (C-D) Growing degree days (GDD) and soil moisture surplus (SMS) at $1$km.}
  \label{fig:ls-data}
\end{figure}

\subsection{Modelling with change of support}%
\label{sub:ls-modelling}

Consider $\rv{LC}(\locc_{i}^{[y]})$ as the land class of interest at cell $i=1,
\dots, n$, and $\rv{X}_j(\locc_{jk}^{[x]})$ as predictor $j=1,\dots,p$ at cell $k$ of the
data source $j$ ($\locc_{jk}^{[x]}$). The \textit{observation model} defines land cover as
a binarisation of a latent process $\{\rv{Z}(\locc_{i}^{[y]})\}$ which depends on a
spatial process $\{\rv{\eta}(\locc_{i}^{[y]})\}$ and an error term
$\varepsilon(\locc_{i}^{[y]}) \sim N(0, \tau^2)$, and it defines the predictors
$\rv{X}_j(\locc_{jk}^{[x]})$ in terms of latent processes $\{\rv{V}_j(\locc_{jk}^{[x]})\}$
and an error term $\xi_j(\locc_{jk}^{[x]}) \sim N(0, \psi^2_j)$.
\begin{align*}
  \rv{LC}(\locc_{i}^{[y]}) & = \left\lbrace\begin{array}{cl}
    1, & \rv{Z}(\locc_{i}^{[y]}) > 0, \\
    0, & \rv{Z}(\locc_{i}^{[y]}) \leq 0.
  \end{array}\right.
  \quad\quad~\text{for}~  i = 1, \dots, n, \\
   \rv{Z}(\locc_{i}^{[y]}) & = \eta(\locc_{i}^{[y]}) +
   \varepsilon(\locc_{i}^{[y]}), \\
  \rv{X}_j(\locc_{jk}^{[x]}) & = \alpha_j + \rv{V}_j(\locc_{jk}^{[x]}) +
  \xi_j(\locc_{jk}^{[x]}),
  \quad~\text{for}~ j = 1,\dots,p ~\text{and}~ k = 1, \dots, m_j.
\end{align*}

Note that $\eta(\locc_{i}^{[y]})$ and $\rv{V}_j(\locc_{jk}^{[x]})$ capture that main
spatial patterns presented in land cover $\rv{LC}(\locc_{i}^{[y]})$ and the predictor
$\rv{X}_j(\locc_{jk}^{[x]})$, respectively. These processes are related to continuous
processes in the \textit{change of support model} such as
\begin{align*}
  \eta(\locc_{i}^{[y]}) & = \int_{\locs \in \locc_{i}^{[y]}} \eta(\locs) d\locs, &
   \ssp(\locc_{i}^{[y]}) & = \int_{\locs \in \locc_{i}^{[y]}} \ssp(\locs) d\locs, &
  \rv{V}_j(\locc_{jk}^{[x]}) & = \int_{\locs \in \locc_{jk}^{[x]}} \rv{V}_j(\locs) d\locs.
\end{align*}

The relationship between the continuous latent processes of interest is defined in the
\textit{latent model}. Considering $LS(\locs)$ as the land suitability at location
$\locs$, then we define
$
   \eta(\locs)  = \beta_0 + \rv{LS}(\locs) + \ssp(\locs) \quad \mbox{and} \quad
   \rv{LS}(\locs)  = \rve{V}^\tr(\locs)\ve{\beta}_1.
$
This implies that land suitability depends on the observed predictors through
$\rve{V}(\locs)$ and influences land cover through $\eta(\locs)$. In this model, one of
our primary interests is to predict the latent continuous process that defines land cover
$\{\rv{Z}(\locs)\}$. However, more importantly, we aim to predict the latent land
suitability $\{\rv{LS}(\locs)\}$ at different supports while considering the associated
uncertainty.


\subsection{Results}%
\label{sub:ls-results}

Urban areas were masked out to avoid assuming that land suitability
determined by the current urban status. This decision aligns
with our objective to characterize land suitability based on climatic, internal soil, and
external soil characteristics, aiming to capture the variability that influences land
suitability beyond the presence of urban development.
We did not include elevation as it was strongly correlated with GDD. The results described
below are based on the model of Section \ref{sub:ls-modelling} including
GDD and SMS as predictors.

\subsubsection{Inference}%
\label{ssub:ls-inference}

\fig \ref{fig:app_inf_latent_process} of SM illustrates the inference of the latent processes in
the model. It can be observed that the mean of the latent process $\{\delta_{v_1}\}$
captures the patterns observed in growing degree days (Figure \ref{fig:ls-data}-C), while
the mean of $\{\delta_{v_2}\}$ captures the patterns observed in soil moisture (\fig
\ref{fig:ls-data}-D). The additional spatial variation not explained by either of the
latent predictors is captured by $\{\delta_w\}$. On the bottom panels, we can also observe
the uncertainty associated with these latent processes, with similar patterns for the
uncertainty of GDD and SMS.
The $95\%$ estimated credible interval for the intercept is $(-0.784,-0.647)$, related to
the global proportion of observed cells with improved grassland. The latent predictor with
respect to growing degree days has a positive effect (CI: $0.881, 1.04$) in defining the
chance of a cell being improved grassland, while the latent predictor related to soil
moisture surplus has also a positive effect (CI: $0.12, 0.274$).

\subsubsection{Prediction}%
\label{ssub:ls-prediction}

Our model can provide predictions at different resolutions, considering the properties
associated with changes in support. The prediction of the
latent
land
cover at 100m, 500m, and 1km can be observed in Figure \ref{fig:app_pred_lc}. All
predictions capture the main patterns observed in the distribution of improved grassland
(Figure \ref{fig:ls-data}-A). Particularly, it can be seen that north and west areas have
lower levels of improved grassland, and uncertainty is higher in those areas.
High levels of latent land cover are predicted even in areas where urban zones
are present because we could easily mask out the urban cells and remove their influence.
This task is not trivial if land-cover data is aggregated. Additionally, the
standard deviation of predictions becomes smaller as we reduce the resolution of the
process.

\begin{figure}[htb]
\begin{center}
  \includegraphics[width=\linewidth]{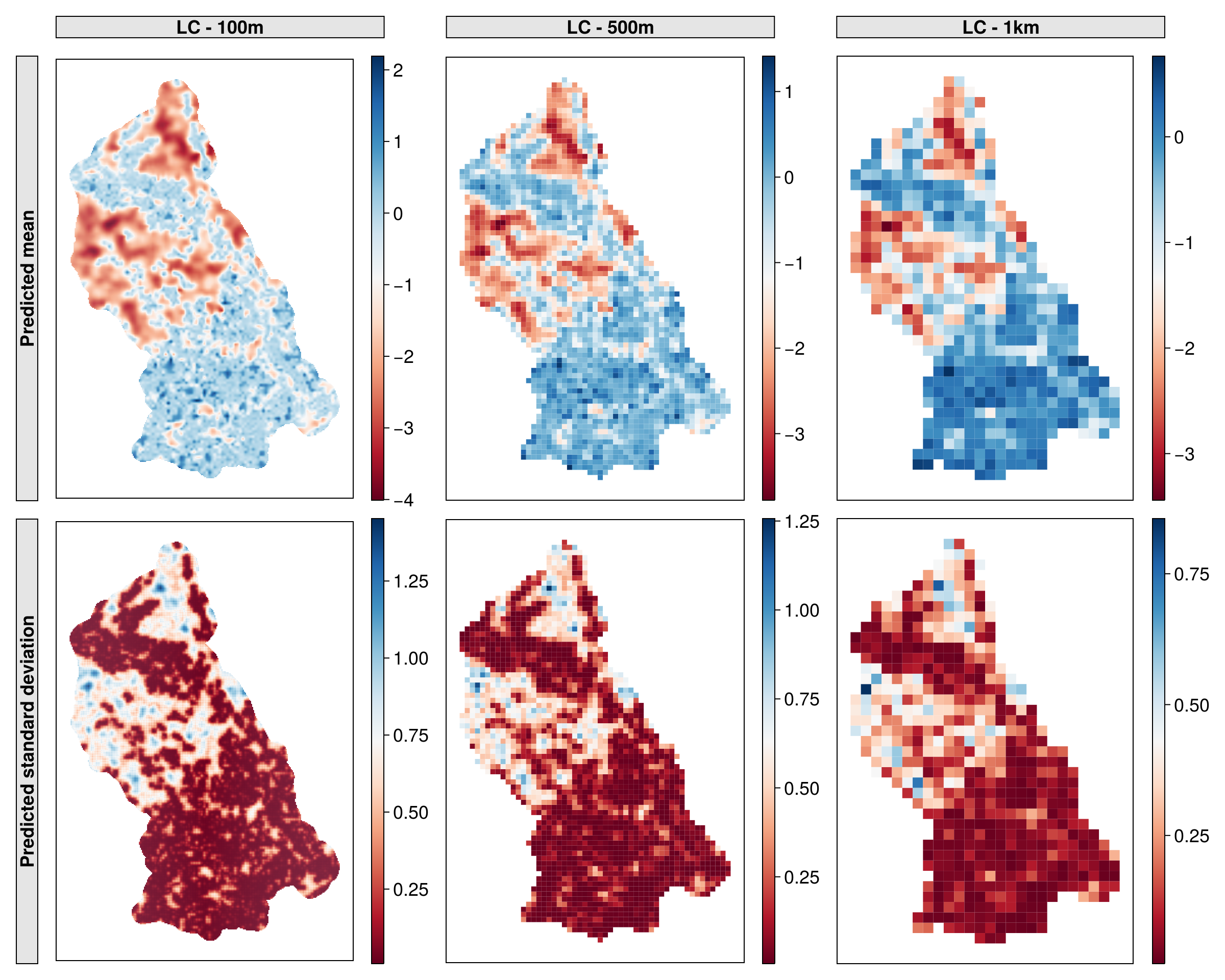}
\end{center}
  \caption{Prediction of latent (continuous) land cover at 100m, 500m and 1km of resolution.}
  \label{fig:app_pred_lc}
\end{figure}

\begin{figure}[htb]
\begin{center}
  \includegraphics[width=\linewidth]{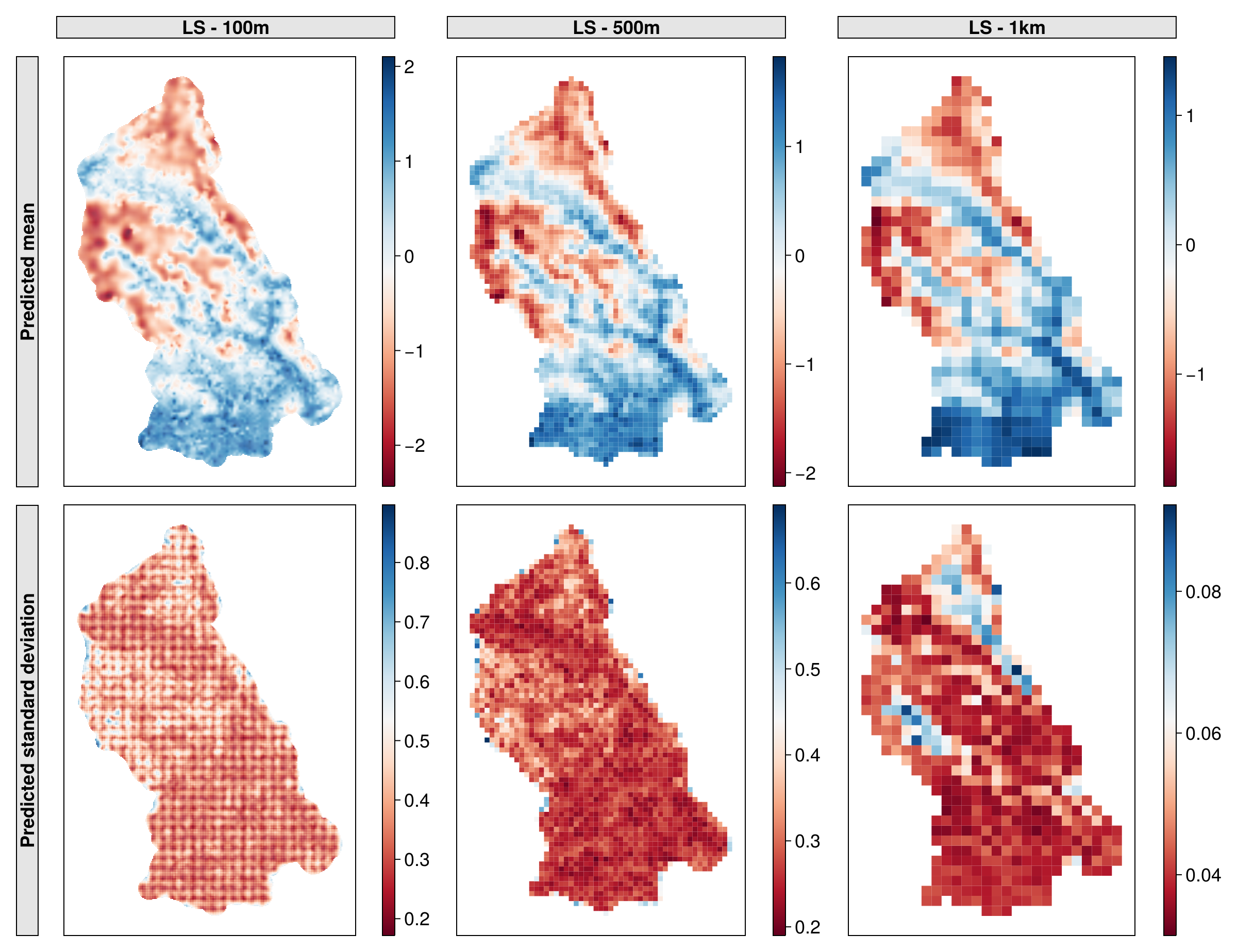}
\end{center}
  \caption{Prediction of land suitability at 100m, 500m and 1km of resolution.}
  \label{fig:app_pred_ls}
\end{figure}

Although our model can predict the latent
land cover, our primary goal is
to predict land suitability, interpreted as the variability explained by intrinsic
variables.
\fig
\ref{fig:app_pred_ls} displays the mean and standard deviation of land suitability
predictions for improved grassland at 100m, 500m, and 1km resolutions.
We observe high levels of suitability in southern areas and low levels
are evident
in the western regions. Moderately low levels of suitability are found in
the north-east. The uncertainty is higher in southern and western areas at
different resolutions. Predictions at different resolutions are
compatible due to the properties of our approach. These resulting surfaces at the desired
resolution can be utilized to predict future scenarios or model land use systems, as
demonstrated in \citet{brown2022agent}.

\section{Discussion}
\label{sec:discussion}


We developed an extension of SLGM that can
accommodate different spatial rectilinear supports for both response and predictor data
sources. Our approach consists of three main components: the \textit{latent Gaussian
model} defining system relationships at a continuous level, the \textit{change of support
model} projecting latent processes to other supports, and the \textit{observation model}
describing data generation mechanisms while acknowledging possible mean biases between
data sources and heteroscedasticity in measurement errors.

In this framework, employing a spatial process with a closed form to
efficiently handle projection to other supports is crucial. We proposed the use of a linear
combination of Gaussian Markov random fields (GMRF) and basis splines, ensuring that the
projection to other supports depends only on the GMRF and the integral of basis
splines over geometries of interest. This approach provides computational advantages, as
the projection does not need approximations on rectilinear grids. The transparent
relationship between the continuous process and the aggregated process is facilitated
through the defined GMRF. The model is computationally efficient, leveraging the
sparse properties of GMRFs and basis splines, enabling inference and prediction with
change of support for large datasets.

Through data simulation, we assessed the efficacy of our model in one and two-dimensional
spaces, considering regular grids, irregular grids, sparse sampling units, and overlapping
sampling units. Consistently, we observed that models neglecting the support introduce
biases in mean and uncertainty quantification and are more prone to over and
underprediction. Conversely, our model captures the main trends of the underlying
processes, providing reliable estimates of uncertainty and exhibiting low levels of
over and underprediction. We further evaluated the adequacy of our model on a more complex
simulated dataset with binary outcomes, different supports for two data sources, and
predictors observed at varying resolutions. Our model not only provided predictions about
the latent processes but also recovered patterns in the latent response process with
reliable uncertainty quantification.

Furthermore, we demonstrated the applicability of our approach in modeling land cover at a
25m resolution in Rhondda Cynon Taf (Wales), utilizing elevation (25m), growing degree
days (1km), and soil moisture surplus (1km). Our model predicted spatial latent processes
behind these variables and provided predictions of latent land cover and land suitability
for improved grassland at various resolutions (100m, 250m, and 1km) along with associated
uncertainties. Our results improved upon previous analyses conducted at a 1km resolution,
addressing differences in support and computational challenges.

Future avenues involve approximating the change of support onto irregular
geometries using a quadtree representation or efficient Monte Carlo
integration.
Additionally, we aim to extend this work to
spatial point
processes, and model jointly with aggregated count data. This
requires
establishing an appropriate relationship between aggregated and continuous processes in
under transformations of latent processes. Efficient sampling
algorithms for the full model are also required due to the presence of several latent
processes, which can be achieved by marginalizing some of the latent processes.

%

\bibliographystyle{apalike}
\bibliography{library,library-additional,library-application}

\newpage

\renewcommand{\appendixpagename}{Supplementary Material}

\begin{appendices}

  In this supplementary material, we delve into theoretical insights and proofs regarding
  the integral of basis splines in \se \ref{sec:app-basis-splines}. We then explore
  common spatial stochastic processes utilized in spatial latent Gaussian models (SLGM) in
  \se \ref{sec:app-slgm}. Additionally, we discuss the theory of SLGM with change of
  support in \se \ref{sec:app-slgm-cos}, covering Bayesian inference and spatial
  prediction. Further details on simulation studies are provided in \se
  \ref{sec:app-simulation-studies}. The MCMC algorithm chains for our case study are
  showcased in \se \ref{sec:app-mcmc-chains}. Lastly, spatial prediction of the latent
  processes for land suitability is presented in Section \ref{sec:app-case-studies}.


\section{Basis splines}
\label{sec:app-basis-splines}

\subsection{Introduction}

\begin{definition}
  Let $\ve{\xi} := (\xi)_1^{l+1}$ be a strictly increasing break sequence, and $\ve{\nu}
  := (\nu)_2^l$ represent homogeneous conditions. Then, $\Pi{<k, \ve{\xi}, \ve{\nu}}$ is
  defined as the space of piecewise $k$-th order polynomial functions.
\end{definition}

A function belonging to the space $\Pi_{<k, \ve{\xi}, \ve{\nu}}$ is characterized by
$k$-order polynomials from $\xi_i$ to $\xi_{i+1}$ and satisfies homogeneous conditions,
specifically $p_i^{(j-1)}(\xi_i) = p_{i+1}^{(j-1)}(\xi_i)$ for $j = 1, \dots, \nu_i$ and
$i = 2, \dots, l$. As an example, when $k = 4$ and ${\nu_i}_{2}^l = 3$, this space is
associated with natural cubic splines.

\textit{Basis splines or B-splines}, defined over a non-decreasing sequence of knots
$\ve{t}$, serve as efficient basis functions for the space of piecewise polynomial
functions $\Pi_{<k, \ve{\xi}, \ve{\nu}}$ \citep{boor2001practical}. The knots $\ve{t}$ are
determined by $\ve{\xi}$ and $\ve{\nu}$ in such a way that $\xi_i$ occurs $k - \nu_i$
times in $\ve{t}$. Note that in the case of natural cubic splines, the breaks $\xi_i$
occur only once in $\ve{t}$.

\begin{definition}
  \label{def:bsplines}
  Considering a non-decreasing sequence of knots $\ve{t}$, the $j$-th \textit{basis
  spline} of order $k$ is defined as
  \begin{equation*}
    B_{j, k}(x) = (t_{j+k} - t_j)[t_j, \dots, t_{j+k}](\cdot-x)^{k-1}_+, ~~ x \in \real,
  \end{equation*}
  which is equal to
  \begin{equation*}
    B_{j, k}(x) = [t_{j+1}, \dots, t_{j+k}](\cdot-x)^{k-1}_+ -
    [t_{j}, \dots, t_{j+k-1}](\cdot-x)^{k-1}_+.
  \end{equation*}
\end{definition}

The operator $[t_j, \dots, t_{j+k}]g$ represents the \textit{k-th divided difference} of
the function $g$, which is defined as the leading coefficient of the $k+1$ order
polynomial function that agrees with $g$ at knots $t_j, \dots, t_{j+k}$. Additionally,
$(\cdot-x)_+ = \max{(\cdot-x), 0}$ represents the truncated function of $(\cdot-x)$ at
$x$. While the above definition is useful for proving properties of basis splines, for
implementation, it is better to use the following representation.

\begin{theorem}
  \label{def:bsplines_recurrency}
  Let $\ve{t}$ be a non-decreasing sequence. The basis splines of order 1 are given by
  \begin{align*}
    B_{j1}(x) & = \left\lbrace\begin{array}{ll}
      1, & \text{if}~ t_j \leq x < t_{j+1}; \\
      0, & \text{otherwise}.
    \end{array}\right.
  \end{align*}
  Additionally, the basis splines of order $k > 1$ can be recursively expressed as
  \begin{equation*}
    B_{jk}(x) = \omega_{jk}(x) B_{j, k-1}(x) + (1 - \omega_{j+1, k}(x)) B_{j+1, k-1}(x),
  \end{equation*}
  where $\omega_{jk}(x) = (x - t_j)/(t_{j+k-1} - t_j)$.
\end{theorem}

The basis spline of order one is directly derived from Definition \ref{def:bsplines}, and
the recursive relationship is established by applying the properties of the $k$-th divided
difference \citep{boor2001practical}.

\begin{definition}
  \label{def:spline_function}
  A \textit{spline function} of order $k$ belonging to the space $\Pi_{<k, \ve{\xi}, \ve{\nu}}$ and
  corresponding knot sequence $\ve{t}$ is defined as a linear combination of basis splines,
  \begin{equation*}
    f(x) = \sum_j\alpha_jB_{jk}(x), ~~ \alpha_j \in \real,
  \end{equation*}
  where $B_{jk}(x)$ is build with respect to the knot sequence $\ve{t}$.
\end{definition}

Out interest in to use spline functions to represent latent processes using stochastic
weights $\{\alpha_j\}$. In the following sections, we analyse the differentiation and
integration of \textit{basis splines} and \textit{spline functions}.

\subsection{Differentiation}

\begin{theorem}
  \label{the:bs_derivative}
  Let $B_{jk}(x)$ be a basis spline of order $k$ starting at knot $t_j$. Then, the derivative can
  be expressed with respect to the basis splines of order $k-1$ as follows:
  \begin{equation*}
    DB_{jk}(x) = (k-1)\left(\frac{B_{j,k-1}(x)}{t_{j+k-1} - t_j} -
    \frac{B_{j+1,k-1}(x)}{t_{j+k} - t_{j+1}}\right).
  \end{equation*}
\end{theorem}

\begin{proof}
  Consider the expanded version of Definition \ref{def:bsplines} for basis splines and
  apply the derivative to both terms:
  \begin{align*}
    DB_{jk}(x) & = [t_{j+1}, \dots, t_{j+k}]D(\cdot-x)^{k-1}_+ -
    [t_{j}, \dots, t_{j+k-1}]D(\cdot-x)^{k-1}_+ \\
    & = -(k-1)[t_{j+1}, \dots, t_{j+k}](\cdot-x)^{k-2}_+ +
    (k-1)[t_{j}, \dots, t_{j+k-1}](\cdot-x)^{k-2}_+.
  \end{align*}
  Notice that, using Definition \ref{def:bsplines}, we can express the $k-1$ divided
  difference as basis splines of order $k-1$, i.e., $[t_{j+1}, \dots,
  t_{j+k}](\cdot-x)^{k-2}_+ = B_{j+1,k-1}(x)/(t_j+k - t_{j+1})$ and similarly for $[t_{j},
  \dots, t_{j+k-1}](\cdot-x)^{k-2}_+ = B_{j,k-1}(x) / (t_{j+k-1} - t_j)$.
\end{proof}

\begin{corollary}
  Let $f(x)$ be a spline function of order $k$ and corresponding knot sequence $\ve{t}$
  with coefficients $\alpha_j \neq 0$ when $r \leq j \leq s$, then the derivative of
  $f(x)$ is
  \begin{align*}
    D\left(\sum_{j=r}^s \alpha_j B_{jk}(x)\right) & = \sum_{j=r}^{s+1} (k-1)
    \left(\frac{\alpha_j - \alpha_{j-1}}{t_{j+k-1} - t_j}\right)
    B_{j,k-1}(x).
  \end{align*}
\end{corollary}

\begin{proof}
  \begin{align*}
    D\left(\sum_{j=r}^s \alpha_j B_{jk}(x)\right) & =
    \sum_{j=r}^s \alpha_j DB_{jk}(x) \\
    & = \sum_{j=r}^s \alpha_j (k-1)\left(\frac{B_{j,k-1}(x)}{t_{j+k-1} - t_j} -
    \frac{B_{j+1,k-1}(x)}{t_{j+k} - t_{j+1}}\right) \tag{by Theorem \ref{the:bs_derivative}}\\
    & = \sum_{j=r}^s \alpha_j (k-1)\left(\frac{B_{j,k-1}(x)}{t_{j+k-1} - t_j}\right) -
    \sum_{j=r+1}^{s+1} \alpha_{j-1} (k-1)
    \left(\frac{B_{j,k-1}(x)}{t_{j+k-1} - t_{j}}\right) \\
    & = \sum_{j=r}^{s+1} (k-1) \left(
    \frac{\alpha_j - \alpha_{j-1}}{t_{j+k-1} - t_j}
    \right) B_{j,k-1}(x).
  \end{align*}
  The last equality holds because $\alpha_{r-1} = 0$ and $\alpha_{s+1} = 0$.
\end{proof}

\subsection{Integration}

\begin{theorem}
  \label{the:bs_integral}
  Let $B_{jk}(t)$ be a basis spline of order $k$ starting at knot $t_j$. Then, the
  integral $\int_{t_j}^x B_{j,k}(t)dt$ for $t_j \leq x \leq t_{j+k}$ can be expressed with
  respect to the basis splines of order $k+1$ as follows:
  \begin{align*}
    \int_{t_j}^x B_{j,k}(t)dt & = \frac{t_{j+k} - t_j}{k} \sum_{i = j}^{s-1} B_{i,k+1}(x),
  \end{align*}
  where $s: t_{s-1} < x < t_s$.
\end{theorem}

\begin{proof}
  Given that the basis spline $B_{jk}(x)$ is non-zero from $t_j$ to $t_{j+k}$, then we can
  evaluate the integral from $t_j$ to an arbitrary value $t_j \leq x \leq t_{j+k}$. Using
  Theorem \ref{the:bs_derivative} to the basis function $B_{j,k+1}(x)$, we obtain that
\begin{align*}
  DB_{j,k+1}(x) & = k\left(\frac{B_{j,k}(x)}{t_{j+k} - t_j} -
  \frac{B_{j+1,k}(x)}{t_{j+k+1} - t_{j+1}}\right), \\
  \int_{t_j}^x DB_{j,k+1}(t) dt & = \frac{k}{t_{j+k} - t_j} \int_{t_j}^x B_{j,k}(t)dt-
  \frac{k}{t_{j+k+1} - t_{j+1}} \int_{t_j}^x B_{j+1,k}(t)dt, \\
  \frac{k}{t_{j+k} - t_j} \int_{t_j}^x B_{j,k}(t)dt & = B_{j,k+1}(x) +
  \frac{k}{t_{j+k+1} - t_{j+1}} \int_{t_{j+1}}^x B_{j+1,k}(t)dt, \\
  \int_{t_j}^x B_{j,k}(t)dt & = \frac{t_{j+k} - t_j}{k} B_{j,k+1}(x) +
  \frac{t_{j+k} - t_j}{t_{j+k+1} - t_{j+1}} \int_{t_{j+1}}^x B_{j+1,k}(t)dt.
\end{align*}
Replacing $\int_{t_{j+1}}^x B_{j+1,k}(t)dt$ using the above equivalence, then
\begin{align*}
  \int_{t_j}^x B_{j,k}(t)dt & = \frac{t_{j+k} - t_j}{k} B_{j,k+1}(x) + \\
  & \quad \frac{t_{j+k} - t_j}{t_{j+k+1} - t_{j+1}} \left(
  \frac{t_{j+k+1}-t_{j+1}}{k}B_{j+1,k+1}(x) +
  \frac{t_{j+k+1} - t_{j+1}}{t_{j+k+2} - t_{j+2}} \int_{t_{j+2}}^x B_{j+2,k}(t)dt
  \right), \\
  & = \frac{t_{j+k} - t_j}{k} B_{j,k+1}(x) + \frac{t_{j+k} - t_j}{k} B_{j+1,k+1}(x) +
  \frac{t_{j+k} - t_j}{t_{j+k+2} - t_{j+2}} \int_{t_{j+2}}^x B_{j+2,k}(t)dt, \\
  & = \frac{t_{j+k} - t_j}{k} \sum_{i = j}^{\infty} B_{i,k+1}(x)
  = \frac{t_{j+k} - t_j}{k} \sum_{i = j}^{s-1} B_{i,k+1}(x).
\end{align*}
For $s$ such as $t_{s-1} < x \leq t_s$. The previous result holds because $B_{i,k+1}(x)
= 0$ for $i \geq s$.
\end{proof}

\begin{corollary}
  \label{the:bs_integral2}
  Let $B_{jk}(t)$ be a basis spline of order $k$ starting at knot $t_j$. Then, the
  integral $\int_{t_i}^x B_{j,k}(t)dt$ for $j \leq i$ and $t_i < x \leq t_{j+k}$ can be expressed with
  respect to the basis splines of order $k+1$ as follows:
  \begin{align*}
    \int_{t_i}^x B_{j,k}(t)dt & = \frac{t_{j+k} - t_j}{k}
    \left(\sum_{r=0}^{s-1}B_{j+r,k+1}(x) - \sum_{r=0}^{i-j-1}B_{j+r,k+1}(t_i)\right),
  \end{align*}
  where $s: t_{s-1} < x < t_s$.
\end{corollary}

\begin{proof}
  This result holds by separating the integral in two parts
\begin{equation*}
  \int_{t_i}^x B_{j,k}(t)dt = \int_{t_j}^x B_{j,k}(t)dt - \int_{t_j}^{t_i} B_{j,k}(t)dt \\
\end{equation*}
and applying Theorem \ref{the:bs_integral} for both sides.
\end{proof}

\begin{corollary}
  Let $f(x)$ be a spline function of order $k$ and corresponding knot sequence $\ve{t}$
  with coefficients $\alpha_j \neq 0$ when $r \leq j \leq n$, then the integral of
  $f(x)$ is
  \begin{align*}
    \int_{t_r}^{x}\left(\sum_{i=r}^n \alpha_i B_{ik}(t)\right)dt
    & = \sum_{i=r}^{s-1} \left(\sum_{j=r}^i\alpha_j(t_{j+k}-t_j)/k\right) B_{i,k+1}(x),
  \end{align*}
  where $s: t_{s-1} < x < t_s$.
\end{corollary}

\begin{proof}
  Noticing that all basis splines $B_{ik}(t)$ for $i > s$, where $s: t_{s-1} < x < t_s$,
  are zero in the interval $[t_r, x]$, then
  \begin{align*}
    \int_{t_r}^{x}\left(\sum_{i=r}^n \alpha_i B_{ik}(t)\right)dt
    & = \int_{t_r}^{x}\left(\sum_{i=r}^{s-1} \alpha_i B_{ik}(t)\right)dt \\
    & = \sum_{i=r}^{s-1}\alpha_i \int_{t_r}^{x}B_{ik}(t)dt \\
    & = \sum_{i=r}^{s-1}\alpha_i \int_{t_i}^{x}B_{ik}(t)dt \\
    & = \sum_{i=r}^{s-1} \alpha_i \frac{t_{i+k} - t_i}{k} \sum_{j = i}^{s-1} B_{j,k+1}(x)
    \tag{Using Theorem \ref{the:bs_integral}} \\
    & = \sum_{i=r}^{s-1} \sum_{j=i}^{s-1}
    \left(\frac{\alpha_i(t_{i+k} - t_i)}{k}\right) B_{j,k+1}(x) \\
    & = \sum_{j=r}^{s-1} \sum_{i=r}^{j}
    \left(\frac{\alpha_i(t_{i+k} - t_i)}{k}\right) B_{j,k+1}(x) \\
    & = \sum_{j=r}^{s-1} \left( \sum_{i=r}^{j}
    \frac{\alpha_i(t_{i+k} - t_i)}{k}\right) B_{j,k+1}(x) \\
    & = \sum_{i=r}^{s-1} \left(\sum_{j=r}^i\alpha_j(t_{j+k}-t_j)/k\right) B_{i,k+1}(x).
  \end{align*}
\end{proof}

\section{Spatial latent Gaussian models}
\label{sec:app-slgm}

\subsection{Spatial Stochastic Processes}%
\label{sub:slgm-spp}

Spatial stochastic processes are essential for the
definition of SLGMs. In this section, we provide more information about Gaussian
processes, Gaussian Markov random fields, and linear transformations of basis functions
which are essential stochastic processes for spatial modelling.

\subsubsection{Gaussian Processes}%
\label{ssub:slgm-ssp-gp}

A spatial Gaussian process (GP) is a stochastic process $\{\ssp(\locs): \locs \in
\mathcal{D} \subset \mathbb{R}^2\}$ with a continuous index set. It is characterised by
the property that any finite subset of random variables follows a multivariate Gaussian
distribution \citep{rasmussen2005gaussian,diggle2007modelbased}. The process
$\gp(\mu(\locs), \kappa(\locs, \locs^*))$ is defined by a mean function $\mu(\locs)$ and a
covariance function, also known as a kernel function, denoted as $\cov{\ssp(\locs)}{\ssp(\locs^*)}
= \kappa(\locs, \locs^*)$ for any pair of locations $\locs,\locs^* \in \mathcal{D}$.
Gaussian processes are widely utilised for modelling spatial data due to their ability to model dependencies
using kernel functions, facilitate statistical inference (both classical and Bayesian),
and provide predictions with uncertainty quantification which is crucial for
decision-making.

\subsubsection{Gaussian Markov Random Fields}%
\label{ssub:slgm-spp-gmrf}

A spatial Gaussian Markov random field (GMRF) is a stochastic process $\{\ssp(\locc):
\locc \in \loccset\}$ with a discrete index set $\loccset = \{\locc_1, \dots, \locc_n\}$
where each $\locc_i$ represents a non-overlapping region. Let the connectivity of the index set
$\loccset$ be defined by the undirected graph $\df{G} = (\df{V}, \df{E})$ with nodes
$\mathcal{V}$ and edges $\df{E} = \{\{i,j\}; i, j \in \df{V}\}$. Then, the random vector
$\bm{\ssp}_c = [\ssp(\locc_1), \dots, \ssp(\locc_n)]^\tr$ follows a multivariate Gaussian
distribution with mean $\ve{\mu}$ and precision matrix $\m{Q} > 0$ such as
\begin{align}
  \label{eq:gmrf_prec_elem}
  Q_{ij} \neq 0 \Longleftrightarrow \{i, j\} \in \df{E} ~ \text{for all} ~ i \neq j.
\end{align}
In cases where $\m{Q}$ is a semidefinite matrix with rank $(n-k)$, the process is known as an
intrinsic Gaussian Markov random field (IGMRF) with density
\begin{align}
  \label{eq:igmrf_density}
  f(\ve{x}) = (2\pi)^{-(n-k)/2}\abs{\m{Q}^*}^{1/2}\exp
  \left(-\frac{1}{2}(\ve{x}-\ve{\mu})^\tr\m{Q}(\ve{x}-\ve{\mu})\right),
\end{align}
where $\abs{\cdot}^*$ denotes the generalised determinant \citep{rue2005gaussian}. The
precision matrix can usually be decomposed as the product of a scalar precision parameter
$\kappa$ and a structure matrix $\m{P}$.

Gaussian Markov random fields are used because they capture the spatial structure based on
the connectivity of the index elements, resulting in a sparse precision matrix $\m{Q}$.
This sparsity leads to lower computational costs compared to Gaussian processes. The
applicability of GMRF is limited to discrete index sets. However, they can also be used to
represent Gaussian processes through their connection with stochastic partial differential
equations \citep{lindgren2011explicit}.

\subsubsection{Linear Combinations of Spatial Basis Functions}%
\label{ssub:slgm-spp-bf}

Another type of continuous-index spatial stochastic process $\{\ssp(\locs): \locs \in
\mathcal{D} \subset \mathbb{R}^2\}$ can be obtained as a linear combination of $q$ spatial
basis functions $b_i(\ve{s})$ such as $\rv{\ssp}(\locs) = \sum_{i=1}^{q} \delta_{i}
b_{i}(\locs).$ Where $\delta_i$, for $i=1,\dots,q$, are the weights associated to the
basis functions. In particular, we focus on the case where the basis functions
$b_i(\locs)$ are locally defined and the spatial structure is introduced by defining
$\{\delta_i\}$ as a GMRF with a two-dimensional regular grid as an index set.
More
specifically, consider the regular grid with $q_1$ and $q_2$ knots corresponding to
\textit{basis splines} $\{B^1_{jk}(\cdot): j=1,\dots,q_1\}$ and $\{B^2_{lk}(\cdot):
l=1,\dots,q_2\}$ of order $k$ for the first and second coordinates, respectively. Then, the
spatial process can be written as
\begin{align}
  \label{eq:bf}
  \rv{\ssp}(\ve{s}) =
  \sum_{j=1}^{q_1}\sum_{l=1}^{q_2} \delta_{jl} B_{jk}^1(s_1) B_{lk}^2(s_2)= \ve{b}^\tr(\locs)\rve{\delta},
\end{align}
where $\rve{\delta}$ is a $q$-dimensional GMRF and $\ve{b}(\locs)$ is a vector containing
the two-dimensional basis functions $b_{jl}(\ve{s}) = B^1_{jk}(s_{1}) B^2_{lk}(s_{2})$
expressed as the tensor product of two sets of one-dimensional basis splines of order $k$.
The advantage of using locally defined basis functions and GMRF is that
$\{\rv{\ssp}(\ve{s})\}$ has also sparse properties leading to lower computational costs.

\section{Spatial latent Gaussian model with change of support}%
\label{sec:app-slgm-cos}

\subsection{Change of Support on Stochastic Processes}%
\label{sub:app-cos-sp}

\subsubsection{Gaussian Processes}%
\label{ssub:app-cos-sp-gp}

\paragraph{Inference:}

Using the properties of a GP for change of support, we can express the likelihood function
for observations with sampling units as points and/or regions to estimate the parameters
$\ve{\theta}$ defined in the mean function $\mu(\locs)$ and the covariance function
$\kappa(\locs,\locs^*)$. However, it is important to acknowledge that in the most commonly
used models, the elements
of the covariance matrix
cannot be obtained analytically. Therefore, approximation methods,
with Monte Carlo integration being the most common, are employed to estimate these
quantities at each iteration of the algorithm to obtain the estimates or to obtain samples
from the posterior distribution.

\paragraph{Prediction:}

Let us consider a setting where we desire predictions at new points and/or regions
$\mathcal{A}^* = \{\locs_1^*, \dots, \locs_{m_\locs}^*, \locc_1^*, \dots,
\locc_{m_\locc}^*\}$. Predictions at these new locations, given a set of observations
at points and/or regions $(\rve{\ssp}_{a^*} \mid \rve{\ssp}_a = \ve{w})$, follow a
multivariate normal distribution with mean vector and covariance matrix
\begin{align*}
  \ve{\mu}_{a^*\mid a} & = \ve{\mu}_{a^*} + \m{\Sigma}_{a^*a}^\tr \m{\Sigma}^{-1}_a
  (\ve{w} - \ve{\mu}_a), \\
  \m{\Sigma}_{a^*\mid a} & = \m{\Sigma}_{a^*} -
  \m{\Sigma}_{aa^*}^\tr \m{\Sigma}^{-1}_a \m{\Sigma}_{aa^*}.
\end{align*}
The covariance matrices $\m{\Sigma}_a$ and $\m{\Sigma}_{a^*}$ are defined as usual,
while the cross-covariance matrix is defined as
\begin{equation}
  \m{\Sigma}_{aa*} =
  \left[\begin{array}{c|c}
    \m{\Sigma}_{\locs\locs^*}  & \m{\Sigma}_{\locs\locc^*} \\ \hline
    \m{\Sigma}_{\locc\locs^*}  & \m{\Sigma}_{\locc\locc^*}
  \end{array}\right],
\end{equation}
which depends of the covariance under change of support.
Therefore, the
predictive distribution at any set of points and/or regions can be derived. However, as in
the inference case, approximation methods are required to compute the mean and covariance
of the predictive distribution with Monte Carlo integration being the most common.

\subsubsection{Linear Combinations of Spatial Basis Functions}%
\label{ssub:app-cos-sp-bf}

\paragraph{Inference:}

Consider a set of point-level $\ve{y}_{\mathcal{C}}$ and region-level
$\ve{y}_{\mathcal{S}}$ observations coming from the following random vectors
$\rve{Y}_\mathcal{S} \mid \rve{\ssp}_\mathcal{S} \sim \df{N}(\rve{\ssp}_\mathcal{S},
\sigma^2_\mathcal{S}\m{I})$ and $\rve{Y}_\mathcal{C} \mid \rve{\ssp}_\mathcal{C} \sim
\df{N}(\rve{\ssp}_\mathcal{C}, \sigma^2_\mathcal{C}\m{I})$. Due to the connection of
$\rve{\ssp}_\mathcal{S}$ and $\rve{\ssp}_\mathcal{C}$ through $\rve{\delta}$, the
conditional joint density of $\rve{Y} = [\rve{Y}_\mathcal{S}^\tr,
\rve{Y}_\mathcal{C}^\tr]^\tr$ is also a normal distribution $\mathcal{N}(\rve{\ssp},
\m{D})$ with mean and covariance
\begin{align}
  \rve{\ssp}  =
  \left[\begin{array}{c}
    \m{B}_\mathcal{C}\\
    \m{B}_\mathcal{S}\\
  \end{array}\right]
  \rve{\delta}, \quad \m{D} = \text{BlockDiagonal}(\sigma^2_\mathcal{S}\m{I},
\sigma^2_\mathcal{C}\m{I}).
\end{align}
 Given that $\rve{\delta}$ is also normally distributed,
the likelihood associated with the point-level and region-level observations is
$\df{N}(\m{0}, \m{B}^\tr\m{Q}^{-1}\m{B})$, where each row of the design matrix $\m{B}^\tr
= [\m{B}_\mathcal{S}^\tr, \m{B}_\mathcal{C}^\tr]^\tr$ is related to the basis functions
evaluated at point-level or region-level. Inference is feasible using maximum likelihood
or Bayesian inference.

\paragraph{Prediction:}

Assuming that we can only observe $\rve{\ssp}_\mathcal{S}$ and $\rve{\ssp}_\mathcal{C}$
through $\rve{Y}_\mathcal{S} \sim \df{N}(\rve{\ssp}_\mathcal{S},
\sigma^2_\mathcal{S}\m{I})$ and $\rve{Y}_\mathcal{C} \sim \df{N}(\rve{\ssp}_\mathcal{C},
\sigma^2_\mathcal{C}\m{I})$, respectively, then the predictive distribution at new
locations and regions $\rve{\ssp}^*$ given the realizations at point-level
$\ve{y}_\mathcal{S}$ and aggregated-level $\ve{y}_\mathcal{C}$ is
\begin{align*}
  \prd{\ve{\ssp}_\mathcal{S}^*, \ve{\ssp}_\mathcal{C}^*
  \mid \ve{y}_\mathcal{S}, \ve{y}_\mathcal{C}} =
  \int \prd{\ve{\delta} \mid \ve{y}_\mathcal{S}, \ve{y}_\mathcal{C}}
  \prd{\ve{\ssp}_\mathcal{S}^*, \ve{\ssp}_\mathcal{C}^* \mid \ve{\delta}}d\ve{\delta}.
\end{align*}
The second term inside the integral is defined by
\begin{align}
  \label{eq:bf_pred_linear_trans}
  \rve{\ssp}^* =
  \left[\begin{array}{c}
    \m{B}_\mathcal{C}^* \\
    \m{B}_\mathcal{S}^* \\
  \end{array}\right]
  \rve{\delta},
\end{align}
while the first term is the posterior distribution of $\rve{\delta}$, which is a normal
distribution with covariance matrix and mean as follows
\begin{align*}
  \m{\Sigma}_{\delta/y_s} & = \left(
  \frac{\m{B}_\mathcal{S}^\tr\m{B}_\mathcal{S}}{\sigma_\mathcal{S}^2} +
  \frac{\m{B}_\mathcal{C}^\tr\m{B}_\mathcal{C}}{\sigma_\mathcal{C}^2} +
  \kappa\m{P}
  \right)^{-1}, \\
  \ve{\mu}_{\delta/y_s} & = \m{\Sigma}_{\delta/y_s} \left(
  \frac{\m{B}_\mathcal{S}^\tr\ve{y}_\mathcal{S}}{\sigma_\mathcal{S}^2} +
  \frac{\m{B}_\mathcal{C}^\tr\ve{y}_\mathcal{C}}{\sigma_\mathcal{C}^2}
  \right).
\end{align*}
We obtain realisations of the prediction by sampling from the posterior of $\rve{\delta}$
and later applying the corresponding linear transformation with \eq
\ref{eq:bf_pred_linear_trans}.

\subsection{Bayesian inference and prediction}%
\label{sub:cos-inference}

We introduce the matrix formulation of the spatial latent Gaussian model
for binary outcomes under a change of support and subsequently, elucidate a Gibbs sampling
scheme to draw samples from the posterior distribution.

\subsubsection{Model}%
\label{ssub:Model}

Consider $\locsubset_k^{[y]}$ as the set of sampling units for the response variable
in the $k-$th data source for $k=1,\dots,K,$ and $\locsubset_j^{[x]},$ for $j=1,\dots,p,$ as the predictor set. The matrix form of the model for binary outcomes under a change of support
can expressed as follows:
\begin{align*}
  \rve{Y}_k & = I(\rve{Z}_k > 0), \quad~\text{for}~ k = 1,\dots,K, \\
  \rve{Z}_k & = b_k\ve{1}_{n_k} + \beta_0\ve{1}_{n_k} +
  \sum_{j=1}^p \beta_j\m{B}_j(\locsubset^{[y]}_{k})\ve{\delta}_{v_j} +
  \m{B}_w(\locsubset^{[y]}_{k})\ve{\delta}_{w} + \ve{\varepsilon}_{k}, \\
  \rve{X}_j & = \alpha_j \ve{1}_{m_j} + \m{B}_j(\locsubset^{[x]}_{j})\ve{\delta}_{v_j} +
  \rve{\xi}_j, \quad~\text{for}~ j = 1,\dots,p.
\end{align*}
Here $\rve{Y}_k$, $\rve{Z}_k$ and $\ve{\varepsilon}_k$ are random vectors of response
variables, auxiliary variables and error terms corresponding to the observations
$\ve{y}_k$ for the $k$-th source. Similarly $\rve{X}_j$ and $\rve{\xi}_j$ are the random
vectors of predictor variables and error terms corresponding to the observed predictor
values $\ve{x}_j$. $\m{B}_j(\locsubset^{[y]}_{k})$, $\m{B}_j(\locsubset^{[x]}_{j})$ and
$\m{B}_w(\locsubset^{[y]}_{k})$ are basis functions evaluated at specified location sets,
which are associated to the latent processes $\rve{\delta}_{v_j} \sim \text{GMRF}(\ve{0},
\kappa_{v_j}\m{P}_{v_j})$ and $\rve{\delta}_{w} \sim \text{GMRF}(\ve{0}, \kappa_w\m{P}_w)$. The
error terms are assumed to be normally distributed with diagonal covariance matrices
such as $\rve{\varepsilon}_{k}\sim \df{N}(\ve{0}, \sigma^2_{y_k}\m{D}_{y_k})$ and
$\rve{\xi}_{j}\sim \df{N}(\ve{0}, \sigma^2_{x_j}\m{D}_{x_j})$. Notice that the auxiliary
variables can also be written as $ \rve{Z}_k = \left[\begin{array}{c|c|c} \ve{1} & \m{A}_k
& \m{V}_k \end{array}\right]\ve{\beta}^* + \m{B}_w(\locsubset^{[y]}_{k})\ve{\delta}_{w} +
\ve{\varepsilon}_{k}, $ where $\m{A}_k$ is a design matrix of dummy variables $k$,
and $\m{V}_k$ is the design matrix of latent predictors for the locations of the $k$-th
source, and $\ve{\beta}^* = [\beta_0, b_1, \dots, b_{K-1}, \beta_1, \dots, \beta_p]^\tr$
the set of coefficient parameters.

The Bayesian model is fully specified by imposing a normal prior distribution for
$\ve{\beta}^*$, $\df{N}(\ve{0}, \m{\Sigma}_\beta)$. An inverse-gamma prior distribution
for the error variances, such as $\sigma^2_{y_k} \sim IG(a_{y_k},b_{y_k})$ and
$\sigma^2_{x_j} \sim IG(a_{x_j},b_{x_j})$. Finally, the scale parameters of the GMRFs are
$\kappa_{v_j} \sim G(a_{v_j},b_{v_j})$. We do not impose a prior distribution on
$\kappa_w$ because it needs to be fixed as explained in \se \ref{ssub:bernoulli-case}.


\subsubsection{Inference}
\label{ssub:inference}

The posterior distribution of the specified model in the previous section is proportional
to:
\begin{align*}
  & \prod_{k=1}^K \prd{\ve{y}_k \mid \ve{z}_k}
  \prod_{k=1}^K \prd{\ve{z}_k \mid \{\ve{\delta}_{v_j}\}, \ve{\delta}_{w}, \ve{\beta}^*, \sigma^2_{y_k}}
  \prod_{j=1}^p \prd{\ve{x}_j \mid \ve{\delta}_{v_j}, \alpha_j, \sigma^2_{x_j}}
  \prod_{j=1}^p \prd{\ve{\delta}_{v_j} \mid \kappa_{v_j}} \\
  &
  \prd{\ve{\delta}_w}
  \prd{\ve{\beta}^*}
  \prd{\{\alpha_j\}}
  \prd{\{\sigma^2_{y_k}\}}
  \prd{\{\sigma^2_{x_j}\}}
  \prd{\{\kappa_{v_j}\}}.
\end{align*}
We can sample from the posterior distribution using Gibbs sampling taking advantage of
conjugacy on the conditional posterior distributions. The posterior conditional
distributions are presented in Section \ref{sec:app-mcmc-chains} of the supplementary
material.

For instance, the conditional
posterior for the latent random vectors $\rve{Z}_k$ is a truncated multivariate normal
distribution:
\begin{align*}
  \prd{\ve{z}_k \mid \cdot} & \propto
  \prd{\ve{y}_k \mid \ve{z}_k}
  \prd{\ve{z}_k \mid \{\delta_{v_j}\}, \ve{\delta}_{w}, \ve{\beta}^*, \sigma_{y_k}^2}
  \tag{truncated normal} \\
  & = \left(\prod_{i=1}^{n_k} \ind{z_{ki} \geq 0}^{y_{ki}} \ind{z_{ki} < 0}^{1 - y_{ki}}\right)
  \df{N}(\m{V}_k^*\ve{\beta}^* + \m{B}_w(\locsubset^{[y]}_{k})\ve{\delta}_{w},
  \sigma^2_{y_k}\m{D}_{y_k}).
\end{align*}

On the other hand, the conditional posterior distributions for the latent GMRFs
$\rve{\delta}_w$ and $\rve{\delta}_{v_j}$ are multivariate normals. For the former, the
posterior uses data available from the $K$ response data sources:
\begin{align*}
  \prd{\ve{\delta}_w \mid \cdot} & \propto
    \prod_{k=1}^K \prd{\ve{z}_k \mid \{\ve{\delta}_{v_j}\}, \ve{\delta}_{w}, \ve{\beta}^*}
    \prd{\ve{\delta}_w} \tag{normal} \\
  \m{\Sigma}_{\delta_{w}\mid \cdot} & = \left(
    \sum_{k=1}^K\sigma^{-2}_{y_k}\m{B}_w(\locsubset_k^{[y]})^\tr\m{D}_{y_k}^{-1}\m{B}_w(\locsubset_k^{[y]})
    + \kappa_w\m{P}_w \right)^{-1}, \\
  \ve{\mu}_{\delta_{w}\mid \cdot} & = \m{\Sigma}_{\delta_{w}\mid \cdot}
    \left(
    \sum_{k=1}^K\sigma^{-2}_{y_k}\m{B}_w(\locsubset_k^{[y]})^\tr\m{D}_{y_k}^{-1}(\ve{z}_k - \m{V}_k^{*}\ve{\beta}^*)
    \right).
\end{align*}
The conditional posterior of $\rve{\delta}_w$ uses information from the $K$ response data
sources but also from the $j$-th predictor:
\begin{align*}
  \prd{\ve{\delta}_{v_j} \mid \cdot} & \propto
    \prod_{k=1}^K \prd{\ve{z}_k \mid \{\ve{\delta}_{v_j}\}, \ve{\delta}_{w}, \ve{\beta}^*, \sigma^2_{y_k}}
    \prd{\ve{x}_j \mid \ve{\delta}_{v_j}, \alpha_j, \sigma^2_{x_j}}
    \prd{\ve{\delta}_{v_j} \mid \kappa_{v_j}} \tag{normal} \\
  \m{\Sigma}_{\delta_{v_j}\mid \cdot} & = \left(
    \sum_{k=1}^K
    \frac{\beta_j^2\m{B}_j(\locsubset_k^{[y]})^\tr\m{D}_{y_k}^{-1}\m{B}_j(\locsubset_k^{[y]})}{\sigma_{y_k}^2} +
    \frac{\m{B}_j(\locsubset_j^{[x]})^\tr\m{D}_{x_j}^{-1}\m{B}_j(\locsubset_j^{[x]})}{\sigma_{x_j}^2} +
    \kappa_{v_j}\m{P}_{v_j}
    \right)^{-1}, \\
  \ve{\mu}_{\delta_{v_j}\mid \cdot} & =
    \m{\Sigma}_{\delta_{v_j}\mid \cdot} \Biggl(
    \sum_{k=1}^K
    \sigma_{y_k}^{-2}\beta_j\m{B}_j(\locsubset_k^{[y]})^\tr \m{D}_{y_k}^{-1}
    \left(\ve{z}_k - b_k\ve{1} - \beta_0\ve{1} - \sum_{q\neq j}
    \beta_q\m{B}_q(\locsubset_k^{[y]})\ve{\delta}_{v_q}
    - \m{B}_w(\locsubset_k^{[y]})\ve{\delta_w}\right) \Biggr. + \\
    & \quad~
    \Biggl.
    \sigma_{x_j}^{-2}\m{B}_j(\locsubset_j^{[x]})^\tr\m{D}_{x_j}^{-1}(\ve{x}_j - \alpha_j\ve{1})
    \Biggr).
\end{align*}
The conditional posterior for the regression coefficients $\rve{\beta}$ is obtained as follows:
\begin{align*}
  \prd{\ve{\beta}^* \mid \cdot} & \propto
    \prod_{k=1}^K \prd{\ve{z}_k \mid \{\ve{\delta}_{v_j}\}, \ve{\delta}_{w}, \ve{\beta}^*, \sigma^2_{y_k}}
    \prd{\ve{\beta}^*} \tag{normal} \\
  \m{\Sigma}_{\beta \mid \cdot} & =
    \left(\sum_{k=1}^K\sigma_{y_k}^{-2}\m{V}_k^{*\tr}\m{D}_{y_k}^{-1}\m{V}_k^* + \m{\Sigma}_\beta^{-1}\right)^{-1}, \\
  \m{\mu}_{\beta \mid \cdot} & = \m{\Sigma}_{\beta \mid \cdot} \left(
    \sum_{k=1}^K \sigma_{y_k}^{-2}\m{V}_k^{*\tr}\m{D}_{y_k}^{-1}(\ve{z}_k - \m{B}_w(\locsubset_k^{[y]})\ve{\delta_w})
    \right).
\end{align*}
Similarly, the conditional posterior distribution for the intercept coefficients of the
predictors are:
\begin{align*}
  \prd{\alpha_j \mid \cdot} & \propto
    \prd{\ve{x}_j \mid \ve{\delta}_{v_j}, \alpha_j, \sigma^2_{x_j}}
    \prd{\alpha_j} \tag{normal} \\
  \sigma^2_{\alpha_j \mid \cdot} & =
    \left(\sigma^{-2}_{x_j}\ve{1}^\tr\m{D}_{x_j}^{-1}\ve{1} +
    \sigma^{-2}_{\alpha_j}\right)^{-1}, \\
  \mu_{\alpha_j \mid \cdot} & = \sigma^2_{\alpha_j \mid \cdot}
    \left(\sigma^{-2}_{x_j}\ve{1}^\tr\m{D}_{x_j}^{-1}
    (\ve{x}_j - \m{B}_j(\loc_j^{[x]})\ve{\delta}_{v_j})
    \right).
\end{align*}

We also take advantage of conjugacy for the error variance parameters $\sigma^2_{y_k}$,
$\sigma^2_{x_j}$, and the scale parameters $\kappa_{v_j}$. For the response error variances
$\sigma^2_{y_k}$, the conditionals are inverse-gamma distributions:
\begin{align*}
  \prd{\sigma^2_{y_k} \mid \cdot} & \propto
  \prd{\ve{z}_k \mid \{\ve{\delta}_{v_j}\}, \ve{\delta}_{w}, \ve{\beta}^*, \sigma^2_{y_k}}
  \prd{\sigma^2_{y_k}} \tag{inverse gamma} \\
  a_{y_k}^* & = a_{y_k} + n_{k}/2, \\
  b_{y_k}^* & = b_{y_k} + \frac{1}{2}
  \left(\ve{z}_k-\m{V}_k^*\ve{\beta}^*-\m{B}_w(\locsubset_k^{[y]})\ve{\delta}_{w}\right)^\tr
  \m{D}_{y_k}^{-1}
  \left(\ve{z}_k-\m{V}_k^*\ve{\beta}^*-\m{B}_w(\locsubset_k^{[y]})\ve{\delta}_{w}\right),
\end{align*}
and similarly for the predictor error variances $\sigma^2_{x_j}$,
\begin{align*}
  \prd{\sigma^2_{x_j} \mid \cdot} & \propto
  \prd{\ve{x}_j \mid \ve{\delta}_{v_j}, \alpha_j, \sigma^2_{x_j}}
  \prd{\sigma^2_{x_j}} \tag{inverse gamma} \\
  a_{x_j}^* & = a_{x_j} + m_{j}/2, \\
  b_{x_j}^* & = b_{x_j} + \frac{1}{2}
  \left(\ve{x}_j-\alpha_j\ve{1}-\m{B}_j(\locsubset_j^{[x]})\ve{\delta}_{v_j}\right)^\tr
  \m{D}_{x_j}^{-1}
  \left(\ve{x}_j-\alpha_j\ve{1}-\m{B}_j(\locsubset_j^{[x]})\ve{\delta}_{v_j}\right).
\end{align*}
Finally, the conditional posteriors for the scale parameters $\kappa_j$ are gamma distributed:
\begin{align*}
  \prd{\kappa_{v_j} \mid \cdot} & \propto
  \prd{\ve{\delta}_{v_j} \mid \kappa_{v_j}}
  \prd{\kappa_{v_j}} \tag{gamma} \\
  a^* & = a + \frac{1}{2}\text{rank}(\m{P}_{v_j}), \\
  b^* & = b + \frac{1}{2}\ve{\delta}_{v_j}^\tr\m{P}_{v_j}\ve{\delta}_{v_j}.
\end{align*}

\subsection{Spatial prediction}%
\label{sub:cos-prediction}

Our main interest for spatial prediction is to infer the latent process
$\{\rv{\eta}(\locs)\}$ at new locations $\locsubset^*$ (point or geometries). Denoting the
associated random vector as $\rve{\eta}(\locsubset^*)$ with values
$\ve{\eta}_{\locsubset^*}$, the posterior predictive distribution is:
\begin{align*}
  \label{eq:pred}
  \prd{\rve{\eta}_{\locsubset^*} \mid \{\ve{y}_k\}, \{\ve{x}_j\}}
    & = \int \prd{\rve{\eta}_{\locsubset^*}, \ve{\beta}, \{\ve{\delta}_{v_j}\}, \ve{\delta}_{w}
    \mid \{\ve{y}_k\}, \{\ve{x}_j\}} d\ve{\beta}d\{\ve{\delta}_{v_j}\}d\ve{\delta}_{w}, \\
    & = \int \prd{\ve{\beta}, \{\ve{\delta}_{v_j}\}, \ve{\delta}_{w} \mid \{\ve{y}_k\}, \{\ve{x}_j\}}
    \prd{\rve{\eta}_{\locsubset^*} \mid \ve{\beta}, \{\ve{\delta}_{v_j}\}, \ve{\delta}_{w}}
    d\ve{\beta}d\{\ve{\delta}_{v_j}\}d\ve{\delta}_{w}.
\end{align*}
Where the first term on the right-hand side is the marginal posterior distribution for
$\ve{\beta}, \{\ve{\delta}_{v_j}\}$ and $\ve{\delta}_{w}$, while the second term is simply a
deterministic relationship: $\rve{\eta}_{\locsubset^*} = \beta_0\ve{1} +
\sum_{j=1}^p\beta_j\m{B}_j(\locsubset^*)\ve{\delta}_{v_j} +
\m{B}_w(\locsubset^*)\ve{\delta}_w$. Samples from the predictive distribution are obtained
by taking the samples for $\ve{\beta}, \{\ve{\delta}_{v_j}\}$ and $\ve{\delta}_{w}$ from the
posterior and using it to generate samples for $\ve{\eta}_{\locsubset^*}$ with the
deterministic expression provided above.


\section{Simulation studies}
\label{sec:app-simulation-studies}

\subsection{One dimensional}%
\label{sub:one_dimensional}

In the one-dimensional experiments, we used the same scenarios and models explained in
Section \ref{sub:two_dimensional}. We observe results consistent with the two-dimensional
case, with some findings that are more pronounced in 2D due to the larger spatial coverage and increased
fluctuations of the simulated latent process. The key findings are summarised in the
following sections.

\subsubsection{Regular grid}
\label{ssub:regular_grid}

For this experiment, we define ${W(\locs)}$ using 100 basis functions of degree 2 and a
GMRF of order 1 with a scale parameter of $\kappa = 1$. The primary differences observed
between the naive model and the support model are described based on the behaviour of the
predicted mean, the uncertainty quantification, and their tendency to either overpredict
or underpredict the true underlying process.

\begin{itemize}
  \item Predicted mean: In the naive model (M1), the predicted mean exhibits abrupt
    changes and tends to overfit the data because it assumes that observations are
    associated with the centroids of the sampling units (\fig \ref{fig:sim1d_regular}-A1). In
    contrast, the support model (M2) yields a more gradual change in the predicted mean,
    as it aims to find a function whose integral over the sampling units closely matches
    the observed data.
  \item Uncertainty quantification: The naive model (M1) tends to underestimate
    uncertainty, particularly near the sampling unit centroids (\fig
    \ref{fig:sim1d_regular}-A1). In contrast, the support model (M2) provides more
    accurate uncertainty estimates, effectively encompassing the true underlying process
    (\fig \ref{fig:sim1d_regular}-A2). These differences stem from the approach of the support model for finding functions whose integrals closely match the data.
  \item Over and under prediction: The naive model (M1) exhibits more frequent values
    where the probability of overprediction is close to 1 and close to 0 (indicating
    underprediction), as seen in \fig \ref{fig:sim1d_regular}-B1. In contrast, the support
    model (M2) features a higher frequency of predictions around the value of 0.5 (\fig
    \ref{fig:sim1d_regular}-B2). This suggests that the support model provides more
    accurate predictions of the underlying true process.
\end{itemize}

\begin{figure}[H]
\begin{center}
  \includegraphics[width=0.9\linewidth]{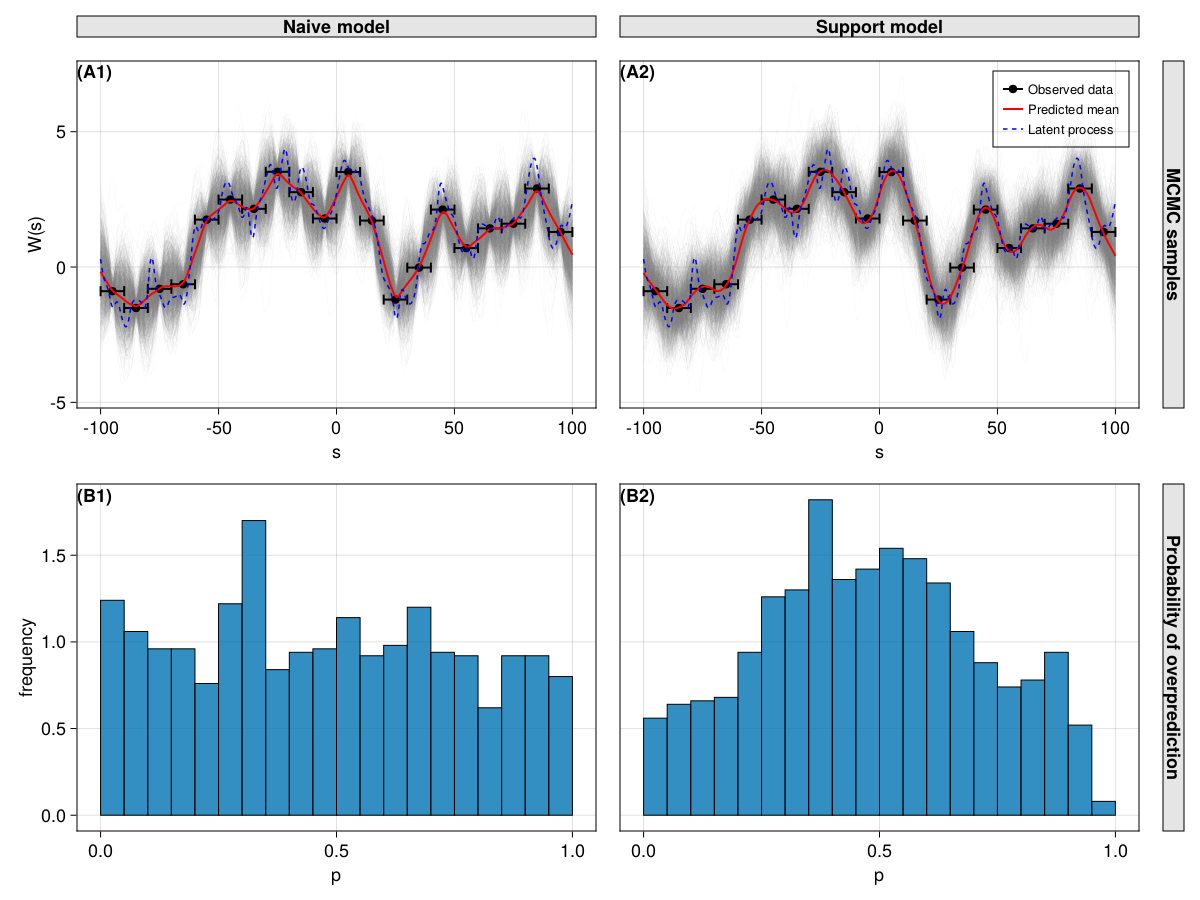}
\end{center}
  \caption{Comparison between the \textbf{naive model} (M1) and the \textbf{support model}
  (M2) when data is observed in a \textbf{regular grid}. Panels A1 and A2 show the
  predictive MCMC samples for the underlying latent process $\{W(s)\}$, while panels B1
  and B2 show histograms of the posterior probability of overpredicting the underlying
  process ($p = \pr{W(s) > w(s) \mid \ve{y}}$).}
  \label{fig:sim1d_regular}
\end{figure}

\subsubsection{Irregular grid, sparse sampling units and overlapping sampling units}
\label{ssub:irregular_grid}

For the irregular grid experiment, the latent process ${W(\locs)}$ was defined with 50
basis functions, while 100 basis functions were used for the experiments with sparse and
overlapping sampling units. In all three experiments, GMRFs of order 1 with a scale
parameter of $\kappa = 1$ were employed. The differences between the models observed in
Figures \ref{fig:sim1d_irregular}, \ref{fig:sim1d_sparse}, and \ref{fig:sim1d_overlapping}
are described below.

\paragraph{Predicted mean:}

In the experiment with irregular grids, the naive model (M1) exhibits overfitting similar
to the previous experiment with regular grids. It is significantly biased by certain
observations due to the lack of accounting for measurement error heteroscedasticity (e.g.
for $\locs \approx 20$ in \fig \ref{fig:sim1d_irregular}-A1). The heteroscedastic model
(M2) performs better by considering observation importance based on support size. In the
case of sparse sampling units, both the naive model (M1) and the heteroscedastic model
(M2) exhibit abrupt changes in the predicted mean, especially evident in sparse locations
like $\locs \in [-50,0]$ in \fig \ref{fig:sim1d_sparse}, panels A1-A2. In an attempt to
fit the data, these models react strongly to the absence of surrounding information. For
overlapping sampling units, the predicted mean of the naive model (M1) and the
heteroscedastic model (M2) does not overfit when units are close, but in mixed sparse and
overlapping regions, they exhibit abrupt changes and overfitting (e.g. for around $\locs
\approx 50$ in \fig \ref{fig:sim1d_overlapping}, panels A1-A2). Across all three
experiments, the support and heteroscedastic model (M3) outperform the naive model (M1)
and the heterocedastic model (M2). It yields smoother variation of the predicted mean,
avoids overfitting, remains unaffected by high-error observations, and accurately captures
primary trends of the underlying true process.

\paragraph{Uncertainty quantification:}

In the case of irregular grids, the naive model (M1) tends to overestimate uncertainty
because it attempts to increase the variability to capture complex data generation mechanisms.
In contrast, the heteroscedastic model (M2) underestimates uncertainty near the sampling
unit centroids for sample units of larger size (e.g., $s \approx 35$ in \fig
\ref{fig:sim1d_irregular}-A2). For sparse sampling units, both the naive model (M1) and
the heteroscedastic model (M2) tend to underestimate uncertainty for locations near the
centroids and locations with sparse data (e.g. $s \approx -50$ in \fig
\ref{fig:sim1d_sparse}, panels A1-A2). In the case of overlapping sampling units, both the
naive model (M1) and the heteroscedastic model (M2) tend to underestimate uncertainty for
locations near the centroids. For all three experiments, the support and heteroscedastic
model (M3) consistently quantifies uncertainty effectively without underestimating it for
locations close to the centroids, as observed in M1 and M2, or overestimating it. M3
provides a more accurate quantification of uncertainty, effectively enveloping the
underlying true process.

\paragraph{Over and under prediction:}

In all the conducted experiments, we observed a higher proportion of posterior
probabilities close to 1 or 0 for the naive model (M1) and the heteroscedastic model (M2),
as shown in panels B1-B3 of Figures \ref{fig:sim1d_irregular}, \ref{fig:sim1d_sparse}, and
\ref{fig:sim1d_overlapping}. In contrast, the support and heteroscedastic model (M3)
consistently exhibits lower levels of overprediction and underprediction when compared to
models M1 and M2, with more frequent probabilities centred around 0.5. This suggests that
M3 provides more accurate predictions of the underlying true process.

\begin{figure}[H]
\begin{center}
  \includegraphics[width=\linewidth]{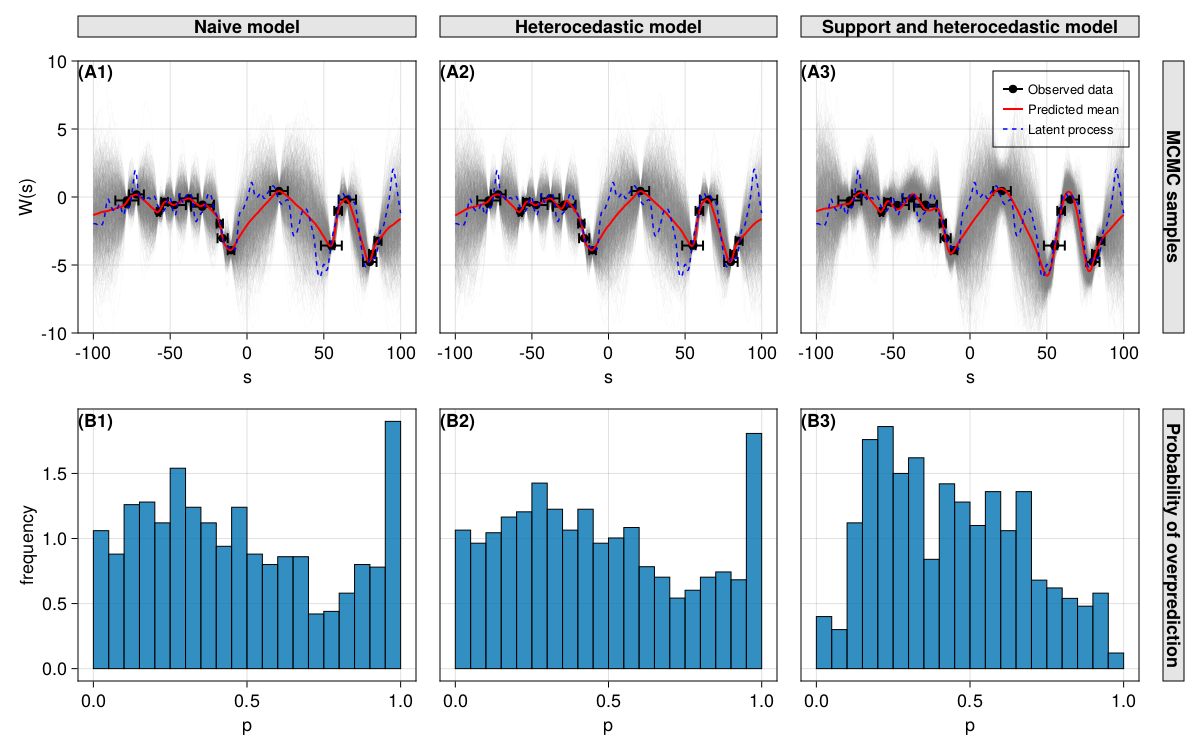}
\end{center}
  \caption{Comparison between the \textbf{naive model} (M1), the \textbf{heterocedastic
  model} (M2), and the \textbf{support and heteroscedastic model} (M3) when data is
  observed in \textbf{overlapping regions}. Panels A1-A3 show the predictive MCMC samples
  for the underlying latent process $\{W(s)\}$, while panels B1-B3 show histograms of the
  posterior probability of overpredicting the underlying process ($p = \pr{W(s) > w(s)
  \mid \ve{y}}$).}
  \label{fig:sim1d_overlapping}
\end{figure}

\begin{figure}[H]
\begin{center}
  \includegraphics[width=\linewidth]{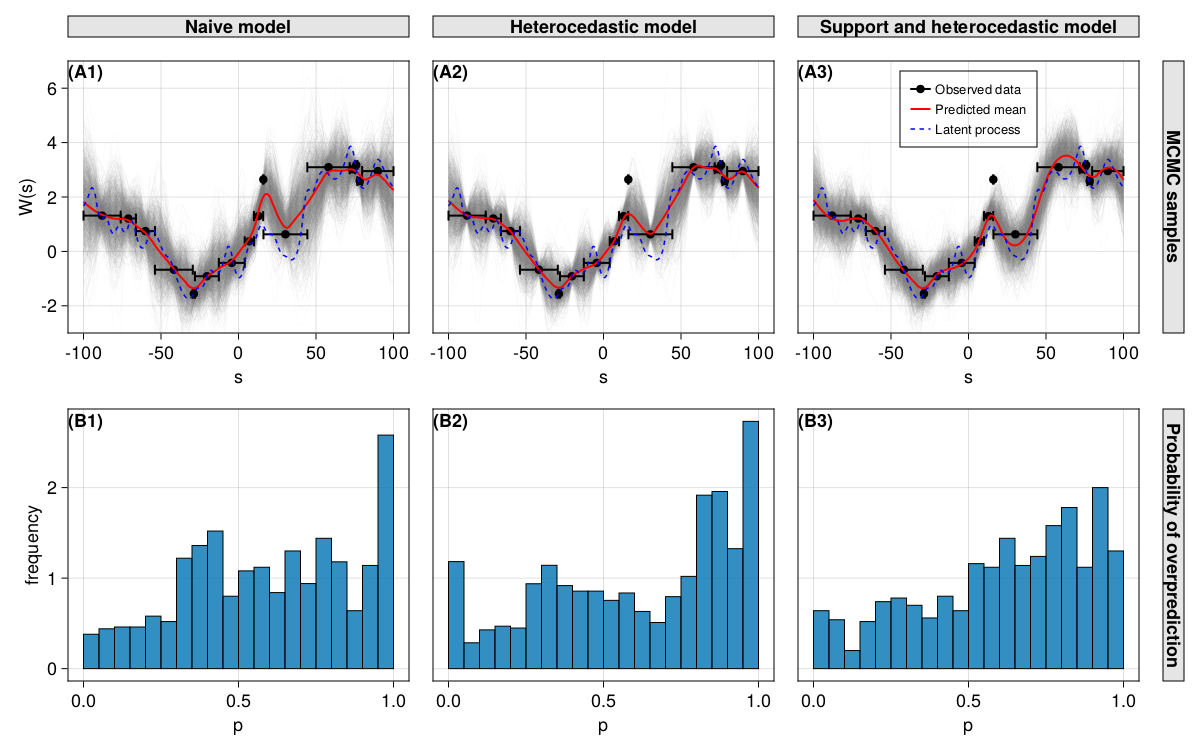}
\end{center}
  \caption{Comparison between the \textbf{naive model} (M1), the \textbf{heterocedastic
  model} (M2), and the \textbf{support and heteroscedastic model} (M3) when data is
  observed in a \textbf{irregular grid}. Panels A1-A3 show the predictive MCMC samples for
  the underlying latent process $\{W(s)\}$, while panels B1-B3 show histograms of the
  posterior probability of overpredicting the underlying process ($p = \pr{W(s) > w(s)
  \mid \ve{y}}$).}
  \label{fig:sim1d_irregular}
\end{figure}

\begin{figure}[H]
\begin{center}
  \includegraphics[width=\linewidth]{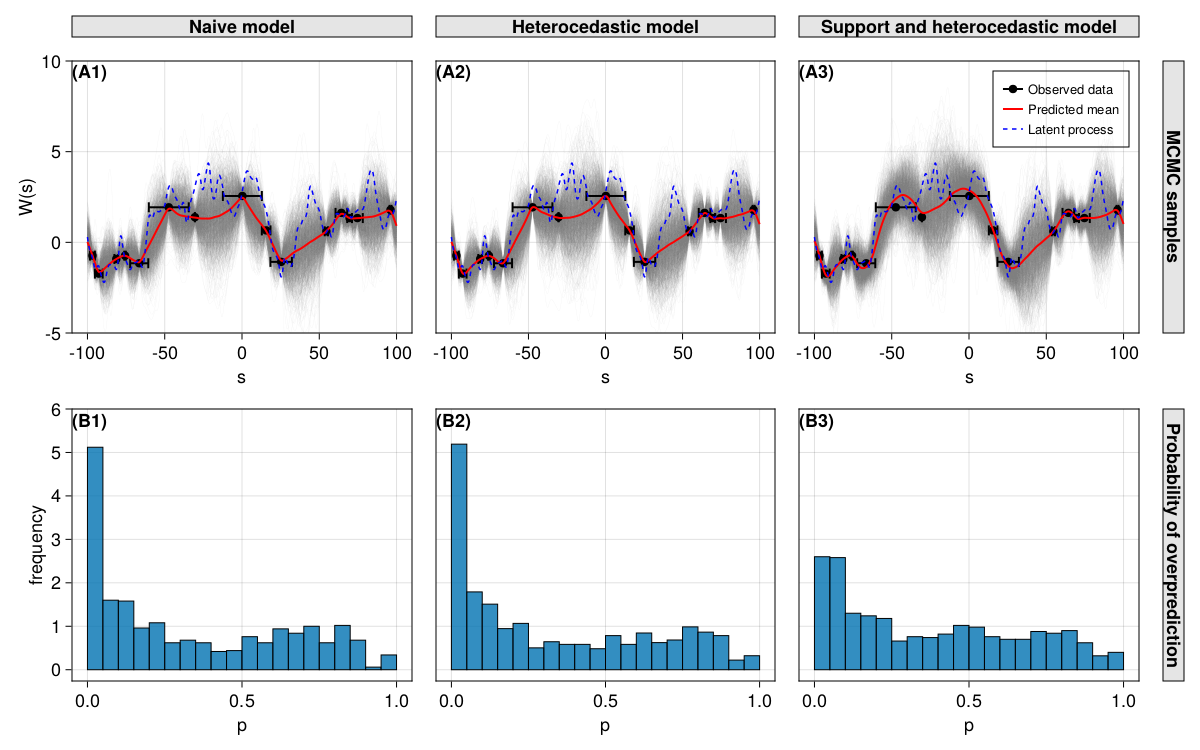}
\end{center}
  \caption{Comparison between the \textbf{naive model} (M1), the \textbf{heterocedastic
  model} (M2), and the \textbf{support and heteroscedastic model} (M3) when data is
  observed in \textbf{sparse regions}. Panels A1-A3 show the predictive MCMC samples for
  the underlying latent process $\{W(s)\}$, while panels B1-B3 show histograms of the
  posterior probability of overpredicting the underlying process ($p = \pr{W(s) > w(s)
  \mid \ve{y}}$).}
  \label{fig:sim1d_sparse}
\end{figure}

\subsection{Two dimensional}%
\label{sub:app_two_dimensional}

\begin{figure}[H]
\begin{center}
  \includegraphics[width=\linewidth]{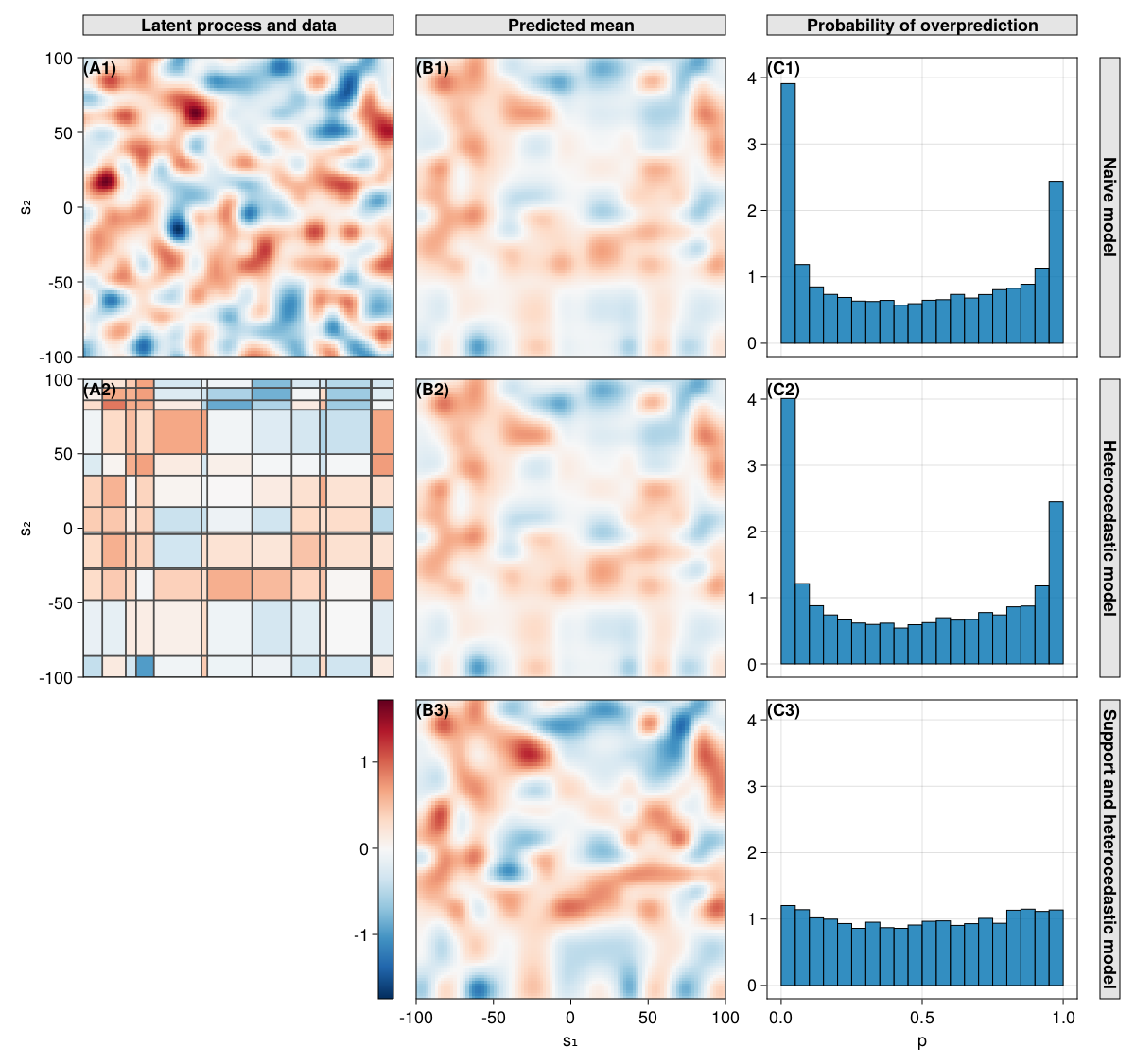}
\end{center}
  \caption{Comparison between the \textbf{naive model} (M1), the \textbf{heterocedastic
  model} (M2), and the \textbf{support and heteroscedastic model} (M3) when data is
  observed in a \textbf{irregular grid}. Panel A1 displays the continuous realization of
  the process of interest, while panel A2 shows the observed data. Panels B1-B3 present
  the predicted mean of the models and panels C1-C3 feature histograms of the posterior
  probability of overpredicting the underlying process ($p = \pr{W(s) > w(s) \mid
  \ve{y}}$).}
  \label{fig:sim2d_irregular}
\end{figure}

\begin{figure}[H]
\begin{center}
  \includegraphics[width=\linewidth]{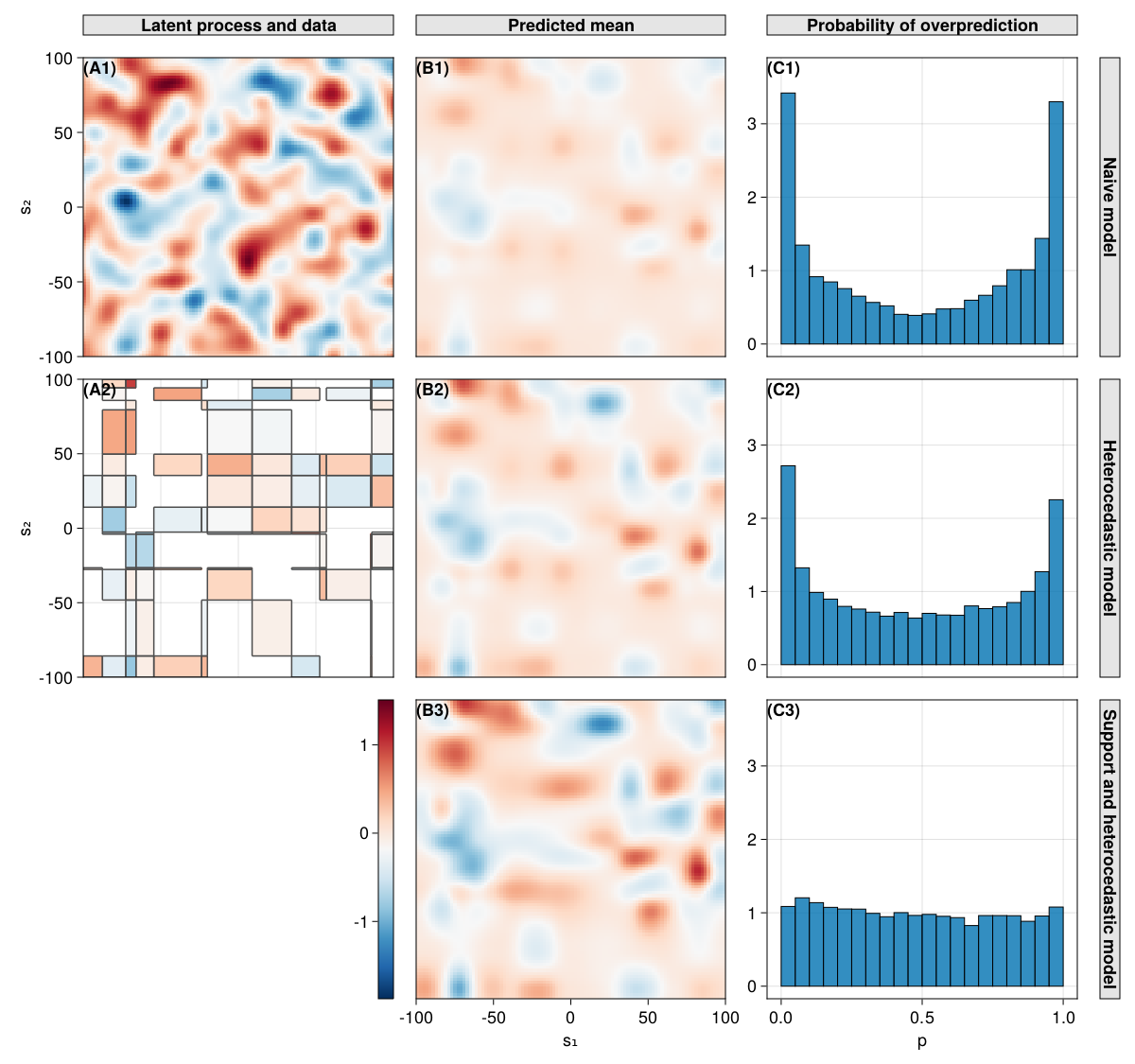}
\end{center}
  \caption{Comparison between the \textbf{naive model} (M1), the \textbf{heterocedastic
  model} (M2), and the \textbf{support and heteroscedastic model} (M3) when data is
  observed in a \textbf{sparse regions}. Panel A1 displays the continuous realization of
  the process of interest, while panel A2 shows the observed data. Panels B1-B3 present
  the predicted mean of the models and panels C1-C3 feature histograms of the posterior
  probability of overpredicting the underlying process ($p = \pr{W(s) > w(s) \mid
  \ve{y}}$).}
  \label{fig:sim2d_sparse}
\end{figure}

\begin{figure}[H]
\begin{center}
  \includegraphics[width=\linewidth]{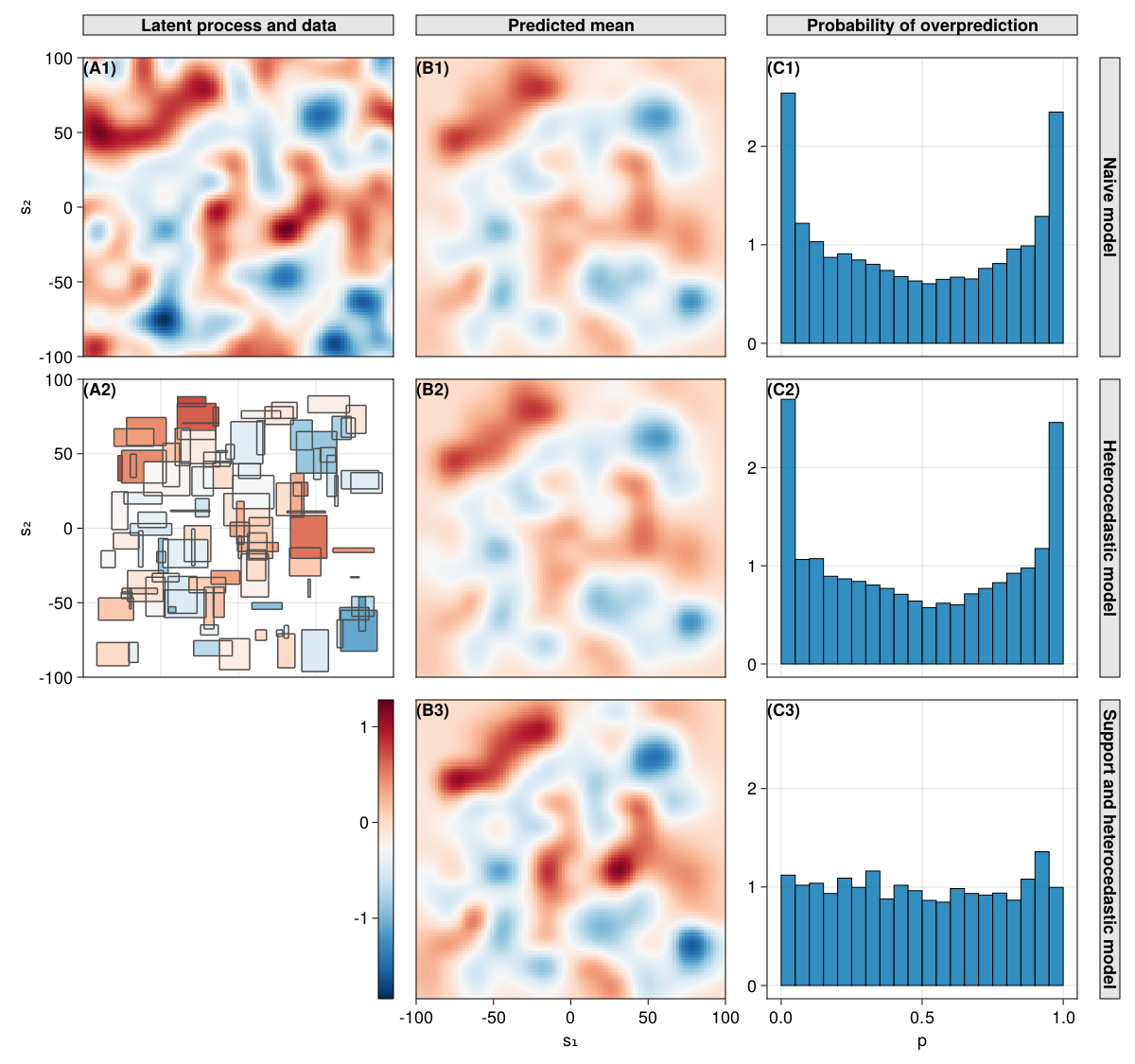}
\end{center}
  \caption{Comparison between the \textbf{naive model} (M1), the \textbf{heterocedastic
  model} (M2), and the \textbf{support and heteroscedastic model} (M3) when data is
  observed in a \textbf{overlapping regions}. Panel A1 displays the continuous realization of
  the process of interest, while panel A2 shows the observed data. Panels B1-B3 present
  the predicted mean of the models and panels C1-C3 feature histograms of the posterior
  probability of overpredicting the underlying process ($p = \pr{W(s) > w(s) \mid \ve{y}}$).}
  \label{fig:sim2d_overlapping}
\end{figure}


\section{MCMC chains}
\label{sec:app-mcmc-chains}


\begin{figure}[H]
\begin{center}
  \includegraphics[width=\linewidth]{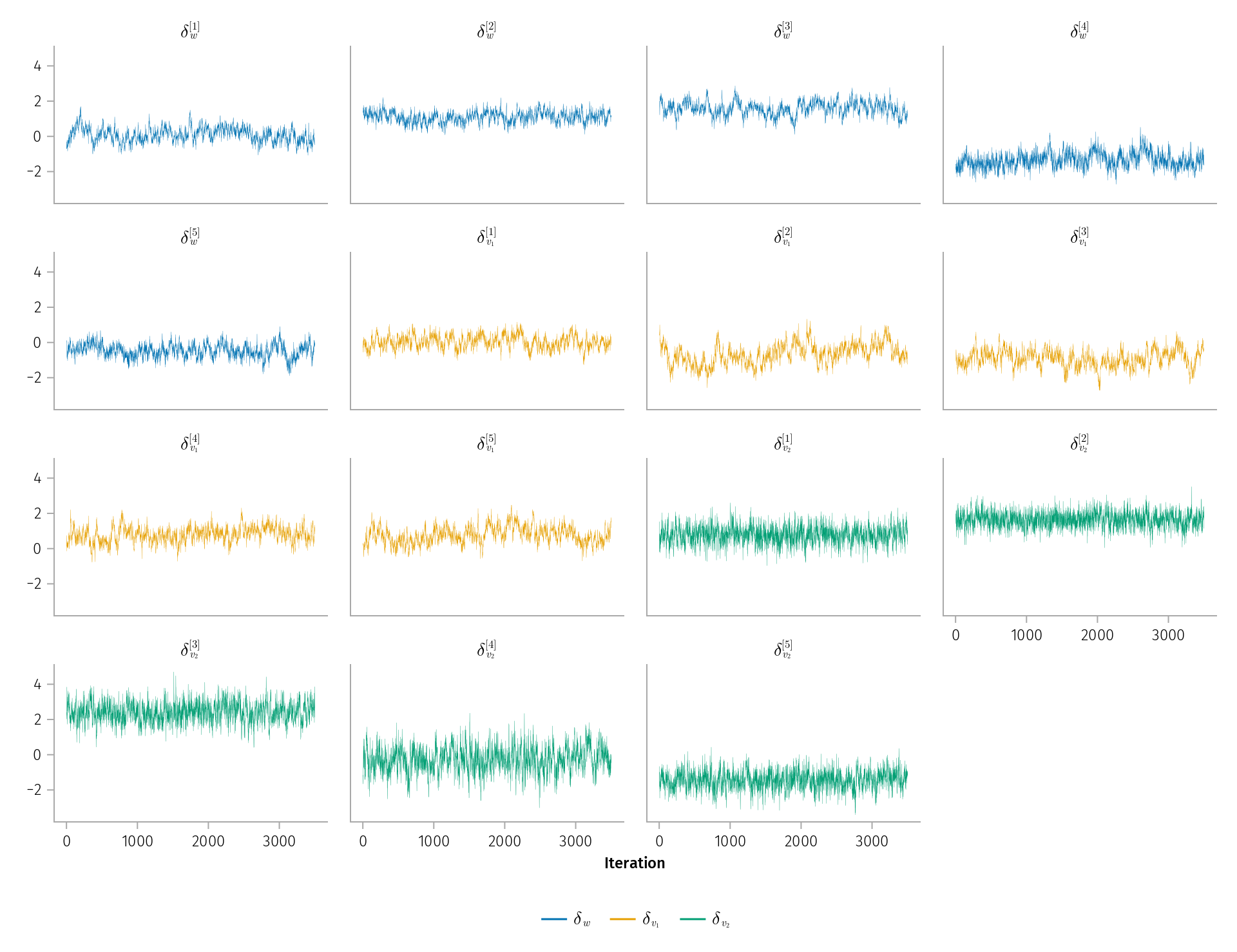}
\end{center}
  \caption{MCMC of the latent fields of land suitability modelling in Rhondda Cynon Taf
  corresponding to 5 randoms cells for each stochastic process.}
  \label{fig:app_chains_latents}
\end{figure}

\begin{figure}[H]
\begin{center}
  \includegraphics[width=\linewidth]{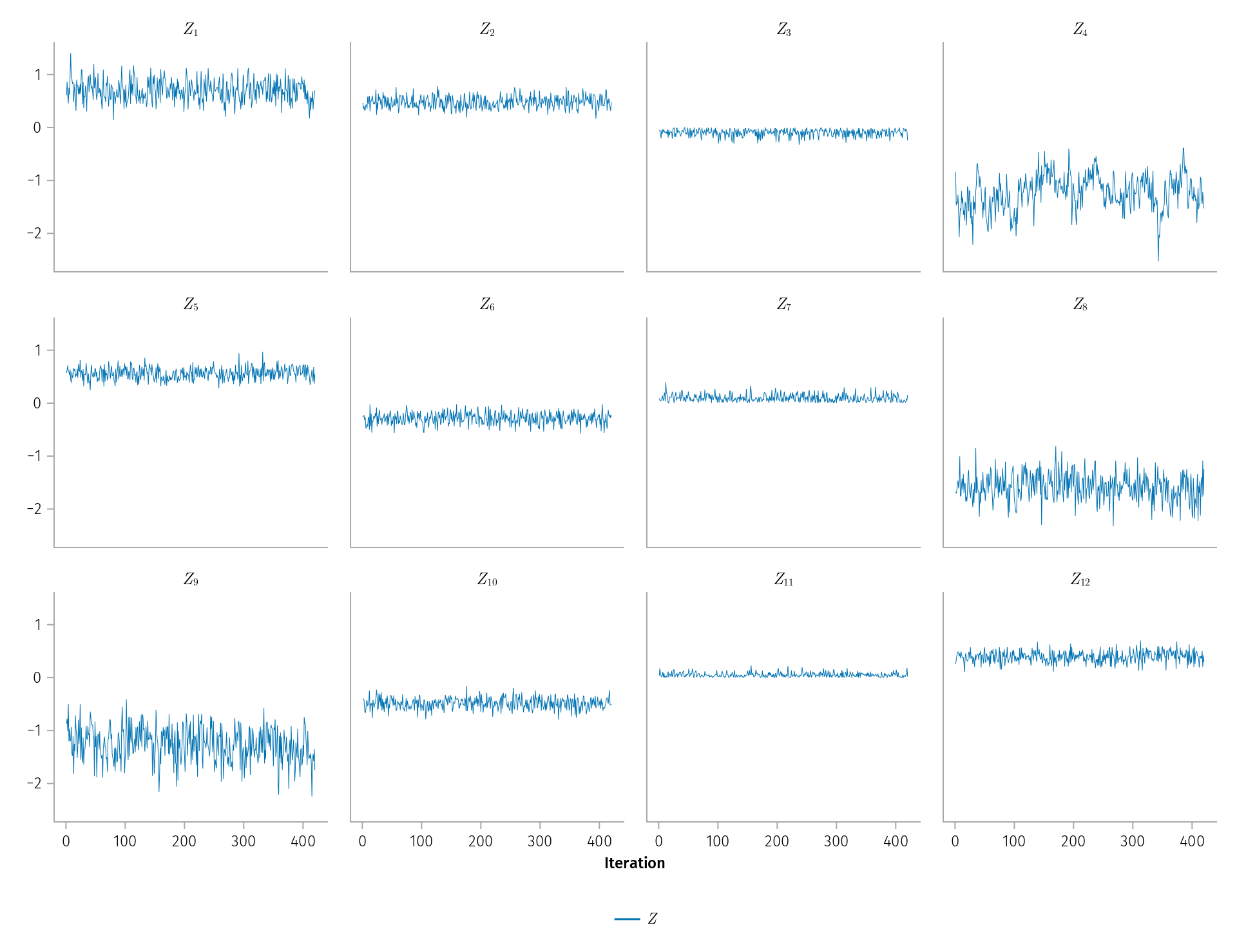}
\end{center}
  \caption{MCMC of the auxiliary process $Z$ of land suitability modelling in Rhondda Cynon Taf
  corresponding to 12 randoms sampling units.}
  \label{fig:app_chains_z}
\end{figure}


\section{Land suitability modelling in Rhondda Cynon Taf}%
\label{sec:app-case-studies}


\begin{figure}[H]
\begin{center}
  \includegraphics[width=\linewidth]{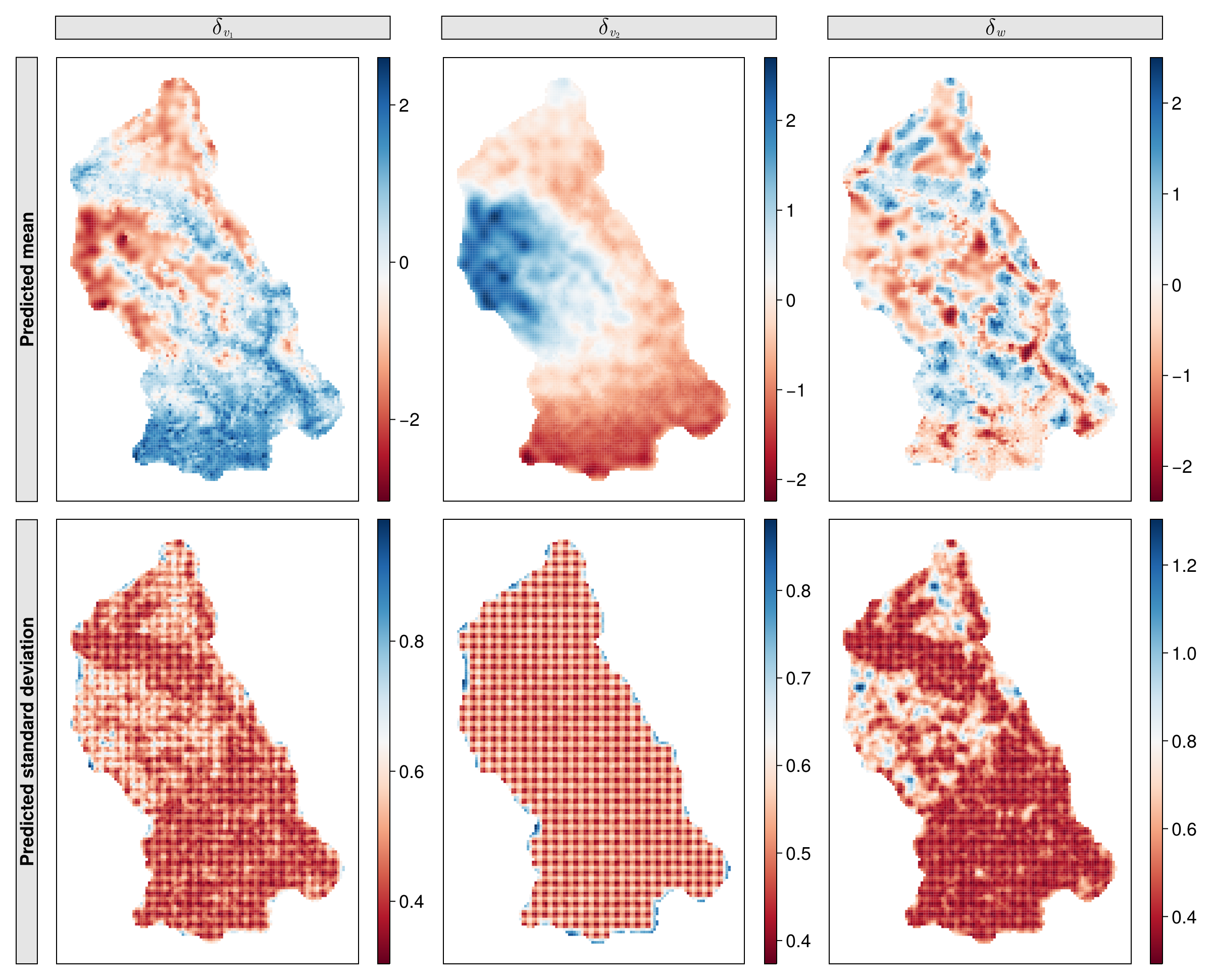}
\end{center}
  \caption{Predicted mean and standard deviation of the latent processes for growing
  degree days ($\delta_{v_1}$), soil moisture surplus ($\delta_{v_1}$), and the residual
  spatial variation ($\delta_{w}$).}
  \label{fig:app_inf_latent_process}
\end{figure}

\end{appendices}

\end{document}